\documentclass[a4paper, pra,twocolumn,showpacs,superscriptaddress,groupedaddress]{revtex4}
\usepackage{ae}
\usepackage{physics}
\usepackage{python}
\usepackage{comment}
\usepackage[T1]{fontenc}
\usepackage[ansinew]{inputenc}
\usepackage{amsmath}
\usepackage{amssymb}
\usepackage[caption=false]{subfig}
\usepackage{multirow}
\usepackage{array}
\usepackage{natbib}
\usepackage[]{graphicx}
\usepackage{wrapfig}
\usepackage{makecell}
\usepackage{epstopdf}
\usepackage{appendix}
\usepackage{mathtools}
\usepackage{color}
\usepackage{graphicx}
\usepackage[colorlinks=true,linkcolor=red,citecolor=blue,urlcolor=blue]{hyperref}
\usepackage{lscape}
\usepackage{amsthm}
\newtheorem{theorem}{Theorem}
\usepackage{tikz}
\usetikzlibrary{arrows.meta,positioning,fit,calc,backgrounds}
\graphicspath{ {./image/} }

\newtheorem{corollary}{Corollary}

\begin{document}

\title{ Pure State Transformations under Block Coherence }
\author{Dipayan Chakraborty}
\email{dipayan.tamluk@gmail.com, dipayan@ssmahavidyalay.org.in}
\affiliation{Department of Mathematics, Sukumar Sengupta Mahavidyalaya, Keshpur, Paschim Medinipur, 721150, West Bengal, India}
\affiliation{Department of Applied Mathematics, University of Calcutta, 92 A.P.C Road, Kolkata, 700009, West Bengal, India}
\author{Priyabrata Char}
\email{priyabrata_ipdf@iitg.ac.in, mathpriyabrata@gmail.com}
\affiliation{Department of Physics, Indian Institute of Technology Guwahati,
Guwahati, 781039, Assam, India}
\author{Indrani Chattopadhyay}
\email{icappmath@caluniv.ac.in}
\affiliation{Department of Applied Mathematics, University of Calcutta, 92 A.P.C Road, Kolkata, 700009, West Bengal, India}
\author{Debasis Sarkar}
\email{dsarkar1x@gmail.com, dsappmath@caluniv.ac.in}
\affiliation{Department of Applied Mathematics, University of Calcutta, 92 A.P.C Road, Kolkata, 700009, West Bengal, India}

\begin{abstract}

    Block coherence provides a natural generalization of standard quantum coherence by treating superpositions across different subspaces as a resource. This work studies deterministic pure-state conversion under three free operations: physically block incoherent operations (PBIO), strictly block incoherent operations (SBIO), and block dephasing covariant incoherent operations (BDCO). For PBIO, we prove that, under a natural nondegeneracy condition on the active Kraus branches, any deterministic conversion from one pure state to another must be implemented by a block incoherent unitary. When the nondegeneracy requirement is removed, the condition becomes more general. It demands that the blockwise action of every active branch reproduce the target block structure with a common proportionality factor across all output blocks. For SBIO and BDCO, we show that deterministic pure-state transformation is completely characterized by the majorization relation between the input and output block probability vectors. The converse proof is constructive, yielding an explicit Kraus representation for every admissible BDCO transformation. In the rank-one limit, these conditions reduce to the known pure-state transformation criteria for physically incoherent operations (PIO), strictly incoherent operations (SIO), and dephasing covariant incoherent operations (DIO) in the standard resource theory of coherence. Using the majorization condition, a maximally block-coherent state with uniform block weights is also identified as a universal pure-state resource under BDCO and SBIO. We have also provided geometric numerical illustrations comparing the state transformation power of  BDCO and DIO for a fixed input state, fixed output state and mutual convertibility scenarios.
\end{abstract}

\maketitle

\section{Introduction}
Quantum coherence, defined relative to a preferred reference basis, is one of the central resources in quantum information theory \cite{Baumgratz_2014_Quantifying,Streltsov_2017_Colloquium}. It underlies interference, quantum algorithms, metrology and several thermodynamic protocols \cite{Giovannetti_2011_AdvancesMetrology,Hillery_2016_DeutschJozsa,Lostaglio_2015_Description, Lostaglio_2015_TimeTranslation,Narasimhachar_2015_LowTemp,Huelga_2013_Vibrations}. In the standard resource theoretic approach, incoherent states are those that are diagonal in the chosen basis. Free operations such as incoherent operations (IO), strictly incoherent operations (SIO), dephasing covariant incoherent operations (DIO) and physically incoherent operations (PIO) are designed so that they do not generate coherence from incoherent inputs \cite{Chitambar_2016_Critical,Yadin_2016_Processes}. For pure states, convertibility under these free operations is often governed by majorization  relations between the corresponding probability vectors \cite{Winter_2016_Operational,Nielsen_1999_Conditions,Du_2015_Conditions}.
In many physical scenarios, the relevant structure is not a set of one dimensional basis states but a decomposition into orthogonal subspaces \cite{Yadin_2016_Processes,Marvian_2016_Speakable,Theurer_2017_Superposition}. Examples include symmetry sectors, energy eigenspaces and coarse-grained measurement outcomes \cite{Marvian_2014_Noether,Lostaglio_2015_TimeTranslation}. In such cases, coherence inside each block may be physically irrelevant, while superpositions between different blocks remain resourceful \cite{Yadin_2016_Processes,Coherence_subspaces}. The resource theory of block coherence \cite{block_coherence, Dey:2019rje,Yadin_2016_Processes,Coherence_subspaces,brub_second,Ren:2021deq} formalizes this idea by fixing an orthogonal decomposition of $\mathcal{H}= \bigoplus_\alpha \mathcal{H_\alpha}$ with projectors $\Pi_\alpha$, and by identifying block diagonal states as free states together with block coherent states as resource states. Like the standard coherence theory, It also introduces several classes of free operations, including MBIO, BIO, SBIO, PBIO, and BDCO \cite{Dey:2019rje}. Each of these operations originated from a particular physical constraint while satisfying the minimum requirement of being free operations. Detailed definitions of block incoherent states and these operations will be discussed in the preliminary section. 

Several aspects of block coherence have subsequently been investigated, including block coherence measures, cohering and decohering powers of quantum channels \cite{Coherence_subspaces}, state transformation, and the relation between block coherence and generalized notions such as $k$-coherence \cite{block_coherence}. Block coherence has also been used to generate entanglement via BIO \cite{kim2021converting}. Furthermore, the framework of block coherence also plays an important role in the study of POVM-coherence \cite{Bischof_2019_POVMCoherence, Bischof_2021_Quantifying}. POVM-coherence is a generalization of standard quantum coherence from orthonormal basis measurements to arbitrary quantum measurements described by POVM. Using the Naimark extension theorem, a POVM can be studied as a projective measurement in a larger Hilbert space, so that POVM-based coherence can be analyzed within the block coherence framework of the extended Hilbert space. The associated free states become block diagonal with respect to the projective measurement obtained from the Naimark extension of the POVM. Thus, block coherence provides a natural framework for studying quantum superposition in situations where experimentally accessible measurements are not rank one projective measurements.\\
Despite these advances, a complete criterion for state transformation under PBIO, SBIO and BDCO has remained an open problem. Although the study presented in \cite{Coherence_subspaces} examines the transformation of pure states under BIO, a thorough characterization of pure-state convertibility under PBIO, SBIO, and BDCO is still missing. Gaining insight into such transformations is crucial both for operational tasks in block coherence and for understanding how superposition between subspaces behaves under different free operations. The free operations impose a hierarchical structure within states based on the resource. If a state $\rho$ can be converted into a state $\sigma$ by a free operation then $\rho$ may be regarded at least as resourceful as $\sigma$. Consequently, any task achievable from $\sigma$ under these free operations is also achievable by $\rho$. Characterizing this hierarchy requires explicit state transformation criteria. It is therefore natural to first establish pure state transformation criteria and then investigate possible extensions to mixed states. Moreover, every resource theory admits a maximally resourceful state, and establishing the state transformation condition is the key step towards identifying this state precisely for block coherence.\\
In this work, we establish necessary and sufficient convertibility criteria for pure states under PBIO, SBIO, and BDCO. We have shown these criteria reduce to the known pure state transformation criteria for PIO, SIO and DIO, respectively, in the rank one limit. Finally, we also provide numerical examples illustrating how nontrivial block structures modify the geometry of the convertibility region and lead to transformation properties beyond the rank-one setting.\\
The remainder of our work is organized as follows. Section \ref{sec:Preliminary} reviews the necessary preliminaries. Section \ref{sec:PBIO} studies pure state transformation under PBIO, while section \ref{sec:BDCO} presents the definition of BDCO and pure state transformation condition under BDCO and SBIO. Section \ref{sec:Maximally block coherence state} identifies the maximally block coherence state under BDCO and SBIO using majorization theory. Section \ref{sec:Numerical Examples} provides geometric numerical illustrations of the transformation regions. Finally, section \ref{sec:Conclusion} concludes the paper.
\section{Preliminaries}
\label{sec:Preliminary}
Let $\mathcal{H}$ be a finite-dimensional Hilbert space with a fixed orthogonal decomposition
$$\mathcal{H}=\bigoplus_{\alpha=1}^m \mathcal{H_\alpha}$$
and let the corresponding mutually orthogonal projectors be
$$\{\Pi_\alpha\}_{\alpha=1}^m, \qquad \sum_\alpha \Pi_\alpha=I.$$

For a pure state $\ket{\psi}$, we write its block decomposition as
$$\ket{\psi}=\sum_{\alpha=1}^m \ket{\psi_\alpha}, \qquad \ket{\psi_\alpha}=\Pi_\alpha \ket{\psi}$$
For $\ket{\psi},\ket{\phi} \in \mathcal{H}$, we define their support sets
$$R = \{\alpha : \ket{\psi_\alpha} \neq 0\}, 
\qquad 
Q = \{\beta : \ket{\phi_\beta} \neq 0\}.
$$

\textbf{Block Incoherent State}: A state $\delta$ is called a block incoherent state \cite{Yadin_2016_Processes,Dey:2019rje,block_coherence,brub_second} if it can be written as 
\begin{equation}
    \label{eq:block incoherent state}
    \delta=\sum_{\alpha}\Pi_\alpha \delta \Pi_\alpha.
\end{equation}
That is, with respect to the fixed block decomposition 
$\{\Pi_\alpha\}_{\alpha=1}^m$.

\textbf{Block Incoherent Unitary}: We  define the block incoherent unitary \cite{Dey:2019rje} $U_j$ as
\begin{equation}
\label{eq: block incoherent unitary}
    U_j= \sum_{\alpha=1}^{m} V_{j, \alpha} \Pi_{\alpha}.
\end{equation}
Here $V_{j,\alpha}: \mathcal{H_\alpha} \rightarrow \mathcal{H}_{\pi_j(\alpha)}$ is a unitary map and $\pi_j$ is a permutation of block labels. Then $V_{j,\alpha}V_{j,\alpha}^{\dagger}=I_{\pi_j(\alpha)}$ and $V_{j,\alpha}^{\dagger}V_{j,\alpha}=I_{\alpha},$ where $I_{\pi_j(\alpha)} \text{ and } I_{\alpha}$ are identity of $\mathcal{H}_{\pi_j(\alpha)}$ and $\mathcal{H}_{\alpha}$ respectively. Note that the permutation map $\pi_j$ will always transfer a block label $\alpha\to\pi_j(\alpha)$ such that dimension of $\mathcal{H_\alpha}$ and $\mathcal{H}_{\pi_j(\alpha)}$ must be same, otherwise no unitary map $V_{j,\alpha}$ exists between them.

\textbf{Example}: We now give a simple example of a block incoherent unitary. Consider a five-dimensional Hilbert space $\mathcal{H}$ and take an orthogonal decomposition as $$\mathcal{H}=\mathcal{H}_1\bigoplus\mathcal{H}_2\bigoplus\mathcal{H}_3,$$ and the corresponding projectors are $$\Pi_1=\ket{0}\bra{0}+\ket{1}\bra{1}, \Pi_2=\ket{2}\bra{2}+\ket{3}\bra{3}, \Pi_3=\ket{4}\bra{4}.$$
\begin{figure}[ht]
\centering
\includegraphics[width=0.45\textwidth]{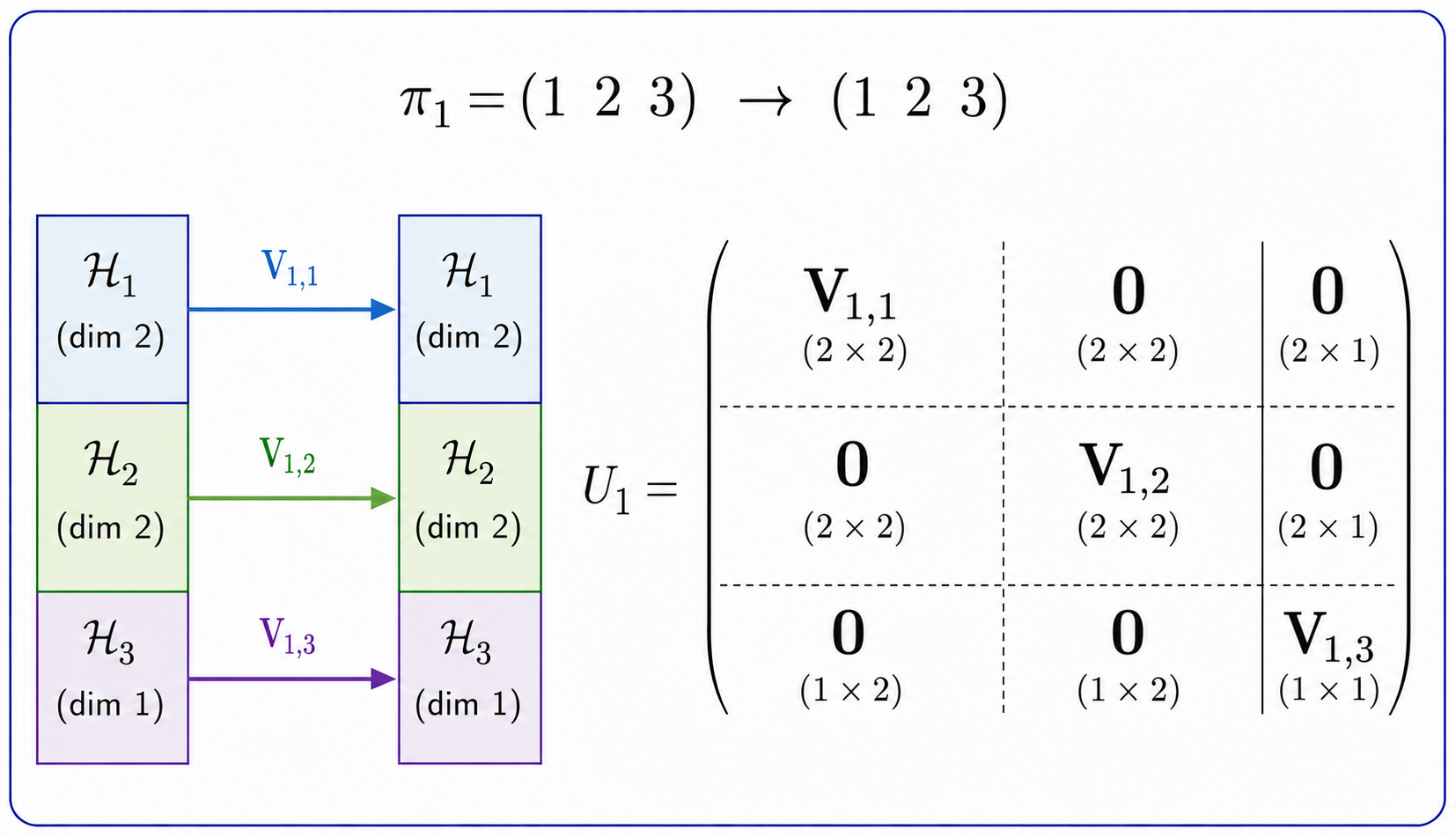}
\caption{Schematic diagram of the block incoherent unitary $U_1$}
\label{fig:Schematic diagram 1}
\end{figure}
Let us consider the block incoherent unitary $U_1 $ as 
$$U_1=V_{1,1}\Pi_1+V_{1,2}\Pi_2+V_{1,3}\Pi_3, \text{ where}$$
$V_{1,1}:\mathcal{H}_1\to\mathcal{H}_{\pi_1(1)}$, $V_{1,2}:\mathcal{H}_2\to\mathcal{H}_{\pi_1(2)}$, $V_{1,3}:\mathcal{H}_3\to\mathcal{H}_{\pi_1(3)}$ are three unitary maps between the corresponding blocks and the associated permutation is $\pi_1=(1,2,3)\to(1,2,3)$. The matrix form of the unitary $U_1$ is given in figure \ref{fig:Schematic diagram 1}.
 Another block incoherent unitary can be $$U_2=V_{2,1}\Pi_1+V_{2,2}\Pi_2+V_{2,3}\Pi_3 \text{ where}$$
$V_{2,1}:\mathcal{H}_1\to\mathcal{H}_{\pi_2(1)}$, $V_{2,2}:\mathcal{H}_2\to\mathcal{H}_{\pi_2(2)}$, $V_{2,3}:\mathcal{H}_3\to\mathcal{H}_{\pi_2(3)}$ are three unitary maps and the corresponding permutation is $\pi_2=(1,2,3)\to(2,1,3)$. The matrix form of the unitary $U_2$ is given in figure \ref{fig:Schematic diagram 2}.
\begin{figure}[ht]
\centering
\includegraphics[width=0.45\textwidth]{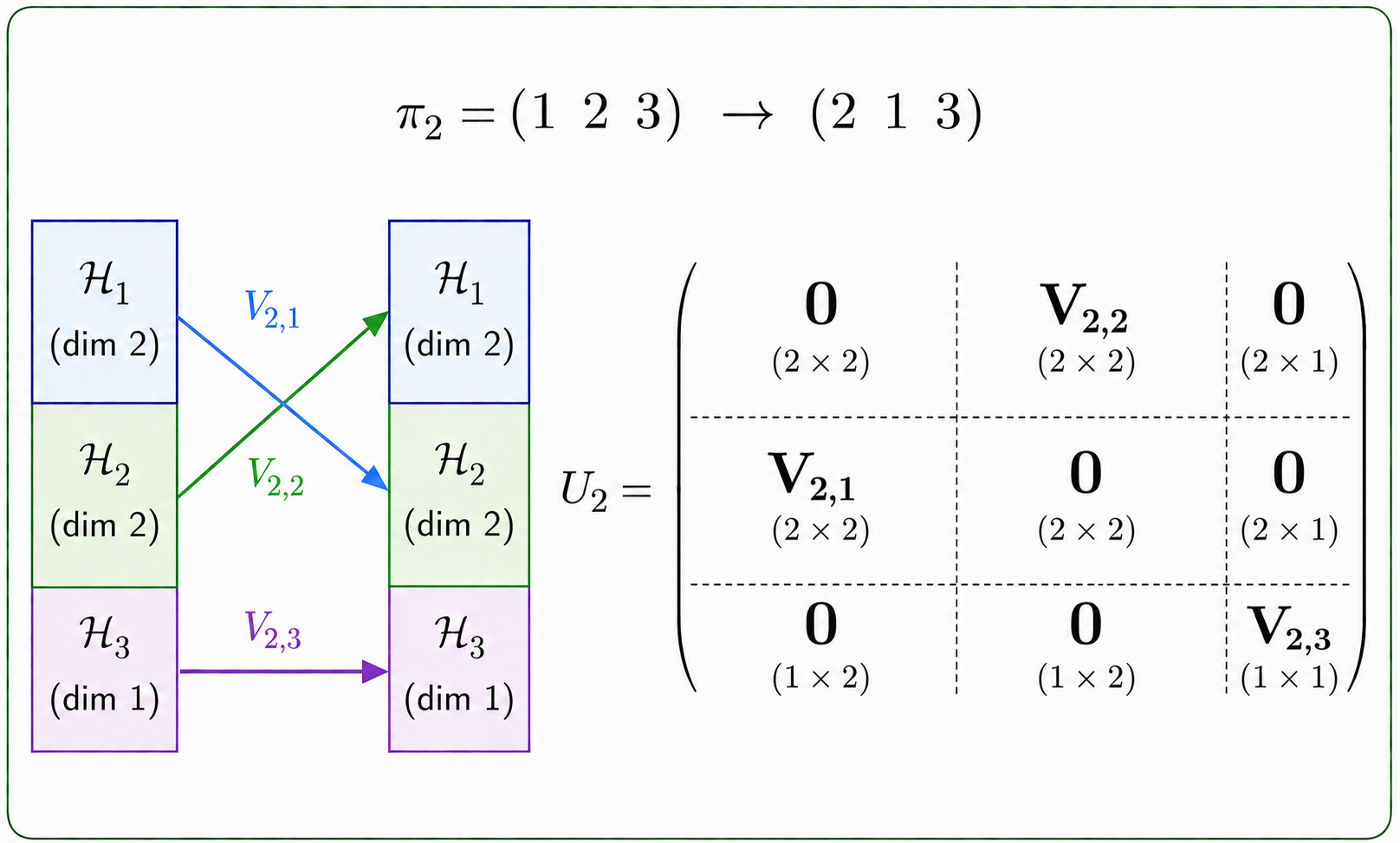}
\caption{Schematic diagram of the block incoherent unitary $U_2$}
\label{fig:Schematic diagram 2}
\end{figure}

\textbf{Block Dephasing Map}: We define the block dephasing map \cite{Yadin_2016_Processes,Dey:2019rje} as
\begin{equation}
\label{eq:block dephasing map}
    \Delta(\rho)=\sum_\alpha \Pi_\alpha \rho \Pi_\alpha.
\end{equation}
 The action of $\Delta$ on a state $\rho$ removes all off-diagonal block terms and maps $\rho$ to a block incoherent state.
 
\textbf{Block Incoherent Operation}: A block incoherent operation (BIO) \cite{Dey:2019rje} is a CPTP map defined as 
\begin{equation}
    \label{eq:BIO} 
    \Lambda(\rho)=\sum_j K_j \rho K_j^\dagger
\end{equation}
such that for all $j$ and block incoherent state $\delta$ $\Delta(K_j \delta K_j^\dagger)=K_j \delta K_j^\dagger.$ The Kraus operator of a BIO can be written in the form. 
\begin{equation}
\label{eq:BIO kraus form}
    K_j=\sum_l \Pi_{f(l)}C_j\Pi_l, 
\end{equation}
where $f:\{1,2,...,m\} \rightarrow \{1,2,...,m\}$ is an index function and $C_j$ is a complex matrix. 

\textbf{Strictly Block Incoherent Operation}: A strictly block incoherent operation (SBIO) \cite{Dey:2019rje} is a CPTP map defined as 
\begin{equation}
    \label{eq:SBIO} 
    \Lambda(\rho)=\sum_j K_j \rho K_j^\dagger
\end{equation}
such that for all $j$ and $\rho$ $\Delta(K_j \rho K_j^\dagger)=K_j \Delta(\rho) K_j^\dagger.$ The Kraus operator of an SBIO can be written in the form
\begin{equation}
\label{eq:SBIO kraus form}
    K_j=\sum_l \Pi_{\pi(l)}C_j\Pi_l, 
\end{equation}
where $\pi$ is a permutation function of the block indices and $C_j$ is a complex matrix. 

\textbf{Physically Block Incoherent Operation}: A physically block incoherent operation (PBIO) \cite{Dey:2019rje} is a CPTP map
$$T(\rho)=\sum_j K_j \rho K_j^\dagger$$
whose Kraus operators are of the form
$$K_j=U_jP_j, \qquad  P_j=\sum_{\alpha \in S_j} \Pi_\alpha$$
where $\{S_j\}_j$ is a partition of $\{1,2,...,m\}$, and each $U_j$ is a block incoherent unitary. 

\textbf{Doubly Stochastic Matrix}: A matrix $T=(t_{ij})_{d\times d}$ is called a doubly stochastic matrix \cite{Marshall_2011_Majorization,Bhatia_1997_MatrixAnalysis} if its elements are non-negative real numbers and the sum of the entries in each row and each column is 1, that is $$\sum_{i=1}^{d}t_{ij}=1,\text{ } \sum_{j=1}^{d}t_{ij}=1\text{ }\forall j,i=1,2,...,d.$$

\textbf{Doubly substochastic Matrix}: A matrix $T=(t_{ij})_{d\times d}$ is called a doubly substochastic matrix \cite{Marshall_2011_Majorization,Bhatia_1997_MatrixAnalysis} if its elements are non-negative real numbers and the sum of the entries of each row and each column are at most 1, that is $$\sum_{i=1}^{d}t_{ij}\leq 1,\text{ } \sum_{j=1}^{d}t_{ij}\leq 1\text{ }\forall j,i=1,2,...,d.$$

\textbf{Column Stochastic Matrix}: A matrix $T=(t_{ij})_{d_1\times d_2}$ is called a column stochastic matrix if its elements are non-negative real numbers and the sum of the entries of each column is 1, that is $$\sum_{i=1}^{d_1}t_{ij}=1,\text{ } \forall j=1,2,...,d_2.$$

\textbf{Majorization}: Let us consider two vectors $\tilde{x}=(x_1,x_2,...,x_d)$ and $\tilde{y}=(y_1,y_2,...,y_d)$. We will say $\tilde{x}$ is majorized \cite{Marshall_2011_Majorization,Uhlmann_1970_Shannon} by $\tilde{y}$ denoted as $\tilde{x}\prec\tilde{y}$ if
\begin{align*}
    x_1^{\downarrow}&\leq y_1^{\downarrow}\\
    x_1^{\downarrow}+x_2^{\downarrow}&\leq y_1^{\downarrow}+y_2^{\downarrow}\\
    ......................&......................\\
    x_1^{\downarrow}+x_2^{\downarrow}+...+x_{d-1}^{\downarrow}&\leq y_1^{\downarrow}+y_2^{\downarrow}+...+y_{d-1}^{\downarrow}\\
    x_1^{\downarrow}+x_2^{\downarrow}+...+x_{d}^{\downarrow}&= y_1^{\downarrow}+y_2^{\downarrow}+...+y_{d}^{\downarrow}
\end{align*}
where $(x_1^{\downarrow},x_2^{\downarrow},...,x_d^{\downarrow})$ and $(y_1^{\downarrow},y_2^{\downarrow},...,y_d^{\downarrow})$ are sorted vectors of $\tilde{x}$ and $\tilde{y}$ respectively, $(x_i^{\downarrow}\geq x_{i+1}^{\downarrow}$ and $y_i^{\downarrow}\geq y_{i+1}^{\downarrow}, \text{ }\forall i=1,2,..,d-1$) of $\tilde{x}$ and $\tilde{y}$ vectors respectively.\\
\textbf{Weak Majorization}
Let us consider two vectors $\tilde{x}=(x_1,x_2,...,x_d)$ and $\tilde{y}=(y_1,y_2,...,y_d)$. we will say $\tilde{x}$ is weakly majorized \cite{Marshall_2011_Majorization,Bhatia_1997_MatrixAnalysis} by $\tilde{y}$ denoted as $\tilde{x}\prec_w\tilde{y}$ if
\begin{align*}
    x_1^{\downarrow}&\leq y_1^{\downarrow}\\
    x_1^{\downarrow}+x_2^{\downarrow}&\leq y_1^{\downarrow}+y_2^{\downarrow}\\
    ......................&......................\\
    x_1^{\downarrow}+x_2^{\downarrow}+...+x_{d-1}^{\downarrow}&\leq y_1^{\downarrow}+y_2^{\downarrow}+...+y_{d-1}^{\downarrow}\\
    x_1^{\downarrow}+x_2^{\downarrow}+...+x_{d}^{\downarrow}&\leq y_1^{\downarrow}+y_2^{\downarrow}+...+y_{d}^{\downarrow}
\end{align*}
where $(x_1^{\downarrow},x_2^{\downarrow},...,x_d^{\downarrow})$ and $(y_1^{\downarrow},y_2^{\downarrow},...,y_d^{\downarrow})$ are sorted vectors of $\tilde{x}$ and $\tilde{y}$ respectively, $(x_i^{\downarrow}\geq x_{i+1}^{\downarrow}$ and $y_i^{\downarrow}\geq y_{i+1}^{\downarrow}, \text{ }\forall i=1,2,..,d-1$) of $\tilde{x}$ and $\tilde{y}$ vectors respectively.

\begin{theorem}
    If $\tilde{x}$ is majorized by $\tilde{y}$, that is $\tilde{x}\prec\tilde{y}$ if and only if there exists a doubly stochastic matrix $T$ such that $ \tilde{x}=T\tilde{y}.$ \cite{Marshall_2011_Majorization}
\end{theorem}

\begin{theorem}
    If $\tilde{x}$ is weakly majorized by $\tilde{y}$, that is $\tilde{x}\prec_w\tilde{y}$ if and only if there exists a doubly substochastic matrix $T$ such that $\tilde{x}=T\tilde{y}  
    .$  \cite{Marshall_2011_Majorization}
\end{theorem}

\section{Pure State Transformation Under PBIO}
\label{sec:PBIO}
\begin{theorem}
\label{thm:pbio}
Let $\ket{\psi}, \ket{\phi} \in \mathcal{H}$, and let
$$T(\rho)=\sum_j K_j \rho K_j^\dagger$$ 
be a PBIO channel. 
We further assume that for every $\beta\in Q$ and for every pair of
distinct active indices $j\neq k$, the union family
$$
\{V_{j,\alpha}\ket{\psi_\alpha} : \alpha\in  A_{j,\beta}\}
\cup
\{V_{k,\alpha}\ket{\psi_\alpha} : \alpha\in  A_{k,\beta}\}
$$
is linearly independent. Then $\ket{\psi}$ can be transformed to $\ket{\phi}$ if and only if there exists a block-incoherent unitary $U_{BI}$ such that
$$
\ket{\phi} = U_{BI}\ket{\psi}.
$$
For each $j$ and $\beta$, $A_{j,\beta} =\{\alpha\in R\cap S_j : \pi_j(\alpha)=\beta\}.$
\end{theorem}

\begin{proof}
We have 
$$
T(\ket{\psi}\bra{\psi}) = \sum_j K_j \ket{\psi} \bra{\psi} K_j^{\dagger}=\ket{\phi}\bra{\phi}.
$$
Since the right-hand side is rank one, every $K_j\ket{\psi}$ must be proportional to $\ket{\phi}$.\\
Therefore for each $j$, there exists $c_j \in \mathbb{C}$ such that
$$K_j \ket{\psi}=c_j \ket{\phi}=\sqrt{r_j} e^{i \theta_j} \ket{\phi}, \;\text{where}\;r_j \geq 0,\;\sum_j r_j=1.$$
Since $T$ is trace preserving, its Kraus operators satisfy
$$\sum_j K_j^\dagger K_j=I.$$
Consequently gives 
$$\sum_j |c_j|^2=1.$$
We now compute left hand side of above expression
\begin{align*}
    K_j \ket{\psi}&=U_j P_j \ket{\psi}\\
    &=U_j \left(\sum_{\gamma \in S_j} \Pi_{\gamma} \right) \left(\sum_\alpha \ket{\psi_{\alpha}}\right)\\
    &=U_j \left( \sum_{\gamma \in S_j } \sum_\alpha \Pi_\gamma \ket{\psi_\alpha}\right)\\
    &=U_j \left( \sum_{\gamma \in S_j } \sum_\alpha \Pi_\gamma \Pi_\alpha \ket{\psi}\right)\\
    &=U_j \left( \sum_{\gamma \in S_j } \sum_\alpha \delta_{\gamma \alpha} \Pi_\alpha \ket{\psi}\right)\\
    &= U_j \left(\sum_{\gamma \in S_j} \Pi_{\gamma} \ket{\psi}\right)\\
    &= \left( \sum_\alpha V_{j,\alpha} \Pi_\alpha \right) \left(\sum_{\gamma \in S_j} \Pi_{\gamma} \ket{\psi}\right) \\
    &=\sum_\alpha \sum_{\gamma \in S_j} V_{j, \alpha} \Pi_{\alpha} \Pi_{\gamma} \ket{\psi}\\
    &=\sum_\alpha \sum_{\gamma \in S_j} V_{j, \alpha} \delta_{\alpha \gamma} \Pi_{\gamma} \ket{\psi}\\
    &=\sum_{\gamma \in S_j} V_{j, \gamma} \ket{\psi_\gamma}
\end{align*}
Thus for every $j$, we have
\begin{equation}
\label{eq 1}
     K_j \ket{\psi}= \sum_{\gamma \in S_j} V_{j,\gamma} \ket{\psi_\gamma}=\sqrt{r_j} e^{i \theta_j} \ket{\phi}.
\end{equation}
Now from (\ref{eq 1}) we can write
\begin{equation*}
    \Pi_{\beta} K_j \ket{\psi} = \Pi_\beta \sum_{\gamma \in S_j} V_{j,\gamma} \ket{\psi_\gamma} =\sum_{\gamma \in S_j: \pi_j(\gamma)=\beta } V_{j,\gamma} \ket{\psi_\gamma}
\end{equation*}
Thus, for every $j$ and $\beta \in Q$
\begin{equation}
\label{eq 2}
    \sum_{\alpha \in S_j: \pi_j(\alpha)=\beta} V_{j, \alpha} \ket{\psi_\alpha}=\sqrt{r_j} e^{i \theta_j} \ket{\phi_\beta}.
\end{equation}
Equation $(\ref{eq 2})$ can be rewritten as 
\begin{equation}
    \label{eq 3}
    \sum_{\alpha \in A_{j, \beta}} V_{j, \alpha} \ket{\psi_\alpha}=\sqrt{r_j} e^{i \theta_j} \ket{\phi_\beta}.
\end{equation}
Since $\pi_j$ is a permutation and for fixed $j, \beta$, the set $A_{j, \beta}$ contains at most one element. Since $j$ is active [A Kraus index $j$ is called active if $K_j \ket{\psi} \neq 0.$], $K_j \ket{\psi} \neq 0$, hence $r_j=|c_j|^2 >0$. Since also $\beta \in Q$, we have $\ket{\phi_\beta} \neq 0$. Therefore the right hand side of the equation (\ref{eq 3}) is nonzero. Hence $A_{j, \beta}$ contains exactly one element.\\
If two different active indices $j \neq k$ exist, then for every $\beta \in Q$, there exist unique indices $\alpha_j(\beta)\in A_{j, \beta}$ and $\alpha_k(\beta) \in A_{k, \beta}$ such that
\begin{equation}
\label{eq 4}
    V_{j, \alpha_j(\beta)} \ket{\psi_{\alpha_j(\beta)}}=\sqrt{r_j} e^{i \theta_j} \ket{\phi_\beta},
\end{equation}
\begin{equation}
\label{eq 5}
    V_{k, \alpha_k(\beta)} \ket{\psi_{\alpha_k(\beta)}}=\sqrt{r_k} e^{i \theta_k} \ket{\phi_\beta}.
\end{equation}
Thus from (\ref{eq 4}) and (\ref{eq 5}) we can write 
\begin{align*}
    V_{j, \alpha_j(\beta)} \ket{\psi_{\alpha_j(\beta)}}&=\sqrt{r_j} e^{i \theta_j} \ket{\phi_\beta}\\
    &=\left( \frac{\sqrt{r_j} e^{i \theta_j}}{\sqrt{r_k} e^{i \theta_k}}\right) \left(\sqrt{r_k} e^{i \theta_k}\right) \ket{\phi_\beta}\\
    &=\left( \frac{\sqrt{r_j} e^{i \theta_j}}{\sqrt{r_k} e^{i \theta_k}}\right) V_{k, \alpha_k(\beta)} \ket{\psi_{\alpha_k(\beta)}}.\\
\end{align*}
Thus $V_{j, \alpha_j(\beta)} \ket{\psi_{\alpha_j(\beta)}}$ and $ V_{k, \alpha_k(\beta)} \ket{\psi_{\alpha_k(\beta)}}$ are collinear which contradicts the linear independence assumption. Hence only one active Kraus index exists.\\
Let $j_*$ denotes such an index. Then
\begin{equation}
\label{eq 6}
K_{j_*} \ket{\psi}=e^{i \theta_{j_*}} \ket{\phi}.
\end{equation}
so, 
$$||K_{j*} \ket{\psi}||=1=||P_{j*} \ket{\psi}||.$$
Since $P_{j*}$ is an orthogonal projector, the equality $$||P_{j_*}\ket{\psi}||=||\ket{\psi}||=1$$ 
implies $P_{j*} \ket{\psi}=\ket{\psi}.$
Thus from \eqref{eq 6} we can write
\begin{equation}
    \label{eq 7}
    U_{j_*} \ket{\psi}=e^{i \theta_{j_*}} \ket{\phi}\;\;\text{and}\;\;\ket{\phi}=e^{-i \theta_{j_*}}  U_{j_*} \ket{\psi}.
\end{equation}
We set $U_{BI}=e^{-i \theta_{j_*}}  U_{j_*} $ and consequently we get $\ket{\phi}=U_{BI} \ket{\psi}.$\\
Conversely, we assume there exists a block incoherent unitary $U_{BI}$ such that 
$$\ket{\phi}=U_{BI} \ket{\psi}.$$
We define a quantum channel by
$$T(\rho)= U_{BI} \rho U_{BI}^\dagger.$$
It is clearly CPTP because it is a unitary map.\\
It remains to show that this is a PBIO. Since $U_{BI}$ is block incoherent unitary, it has the form 
$$U_{BI}=\sum_\alpha V_\alpha \Pi_\alpha,$$
where $$V_\alpha : \mathcal{H_\alpha} \rightarrow \mathcal{H_{\pi(\alpha)}}$$
is a unitary map between blocks and $\pi$ is a permutation of block labels.\\
Now we choose a single Kraus operator
$$K_1=U_{BI}.$$
This has the PBIO form 
$$K_1= U_1 P_1,$$
$$\text{where, }U_1=U_{BI}\;\text{and} \;\;P_1=I=\sum_{\alpha=1}^m \Pi_\alpha.$$
Equivalently, we take the subset $S_1=\{1,2,...,m\}.$
Then $P_1=\sum_{\alpha \in S_1} \Pi_\alpha=I.$
$$\text{Therefore,   } K_1=U_{BI}I=U_{BI}.$$
Hence $T$ is a PBIO channel and satisfies $T(\ket{\psi}\bra{\psi})=\ket{\phi}\bra{\phi}$. This completes the proof.
\end{proof}

Theorem \ref{thm:pbio} shows that, under the nondegeneracy assumption, a PBIO channel can convert $\ket{\psi}$ into $\ket{\phi}$ only when $\ket{\phi}$ is obtained from $\ket{\psi}$ by a block incoherent unitary. We call this a nondegeneracy assumption because it rules out linear dependence among the active branch contributions to each output block, preventing different Kraus branches from producing the same block direction in a degenerate way. Equivalently this transformation reduces to a deterministic reshuffling of the block subspaces combined with internal rotations within each block.\\
This assumption has a clear physical meaning. It prevents two distinct Kraus branches of the channel from producing the same output block component through parallel contributions. Consequently, multiple measurement branches cannot combine to generate the target pure state. The transformation is therefore implemented by a single Kraus operator. Hence, for this deterministic pure state conversion, PBIO cannot change the intrinsic block coherence structure of the state. It can only redistribute it through a symmetry of the block decomposition.

\begin{corollary}
    Without nondegeneracy condition,$$ T(|\psi\rangle\langle\psi|)=|\phi\rangle\langle\phi|.$$ holds if and only if $$\sum_{\alpha \in A_{j, \beta}}V_{j, \alpha} \ket{\psi_\alpha}=\sqrt{r_j} e^{i \theta_j} \ket{\phi_\beta}= c_j \ket{\phi_\beta}$$ holds for every $j$ and every $\beta$
\label{Corollary 1}
\end{corollary}

\begin{proof}
    The necessary part condition follows from equation ($\ref{eq 3}$).\\
    For sufficiency, we assume that the given condition $$ \sum_{\alpha \in A_{j, \beta}}V_{j, \alpha} \ket{\psi_\alpha}=\sqrt{r_j} e^{i \theta_j} \ket{\phi_\beta},$$ 
    holds for every $j$ and for every $\beta$.\\
    Summing over all $\beta$, and using that $\{ A_{j, \beta}\}_\beta$ partitions $R \cap S_j$, we get
    \begin{equation}
    \label{eq 8}
        \sum_{\beta=1}^m \sum_{\alpha \in A_{j, \beta}} V_{j, \alpha} \ket{\psi_{\alpha}}=\sum_{\alpha \in {R \cap S_j}}  V_{j, \alpha} \ket{\psi_{\alpha}}=\sum_{\beta=1}^m \sqrt{r_j} e^{i \theta_j} \ket{\phi_\beta} 
    \end{equation}
    Now $$\sum_{\alpha \in S_j}  V_{j, \alpha} \ket{\psi_{\alpha}}=\sum_{\alpha \in R \cap S_j} V_{j, \alpha} \ket{\psi_{\alpha}}+\sum_{\alpha \in S_j\setminus R} V_{j, \alpha} \ket{\psi_{\alpha}}.$$
  Since the remaining terms with $\alpha \in S_j\setminus R$ vanish, it follows from \eqref{eq 8} that
    
    $$\sum_{\alpha \in S_j} V_{j, \alpha} \ket{\psi_{\alpha}}=\sum_{\beta=1}^m \sqrt{r_j} e^{i \theta_j} \ket{\phi_\beta} = c_j \sum_{\beta = 1}^m \ket{\phi_\beta}
    $$
    Hence
    $$\sum_{\alpha \in S_j} V_{j, \alpha} \ket{\psi_{\alpha}}=c_j \ket{\phi}$$
    From equation ($\ref{eq 1}$) we can write the above equation as
    $$K_j \ket{\psi}=c_j \ket{\phi}$$
    Thus we have
    $$T(\ket{\psi}\bra{\psi})=\sum_j K_j \ket{\psi} \bra{\psi} K_j^\dagger=\sum_j |c_j|^2 \ket{\phi} \bra{\phi} $$
    But $\sum_j |c_j|^2=1$ and this completes the proof.
\end{proof}

    Corollary \ref{Corollary 1} drops the nondegeneracy assumption and provides the most general condition for pure state conversion under PBIO. The condition states that in each branch $j$, the rotated block components of $\ket{\psi}$ must sum to a vector proportional to the target block component $\ket{\phi_\beta}$, with a proportionality constant independent of the output block $\beta$. Thus each Kraus branch must independently reproduce the correct block structure of the target state. In this Kraus representation, different active branches contribute with weights $|c_j|^2$ and these weights sum to one.

\begin{corollary}
    If all block projectors $\Pi_\alpha$ have rank one, then the pure state transformation condition under PBIO reduces exactly to the pure state transformation condition under PIO.
\label{Corollary 2}
\end{corollary}

\begin{proof}
    When every $\Pi_\alpha$ has rank one, both $\mathcal{H}_\alpha$ and $\mathcal{H}_{\pi_j(\alpha)}$ are one dimensional. Now any block incoherent unitary $U_j$, has the form 
    $$U_j= \sum_\alpha V_{j, \alpha} \Pi_\alpha$$
    where
    $$V_{j,\alpha}: \mathcal{H_\alpha} \rightarrow \mathcal{H}_{\pi_j(\alpha)}$$ is a unitary but both $\mathcal{H}_\alpha$ and $\mathcal{H}_{\pi_j(\alpha)}$ are one dimensional. Therefore $V_{j, \alpha}$ is only a phase map. Thus, there exists $\omega_{j, \alpha} \in \mathbb{R}$ such that 
    $$V_{j, \alpha} \ket{\alpha}=e^{i \omega_{j, \alpha}} \ket{\pi_j (\alpha)}$$
    Consequently,
    $$U_j \ket{\alpha}= e^{i \omega_{j, \alpha}} \ket{\pi_j (\alpha)}$$
    Thus each $U_j$ is an incoherent unitary.\\
    Also, the projector $$P_j= \sum_{\alpha \in S_j} \Pi_\alpha$$ is simply the projector onto a subset of the incoherent basis.\\
    Hence each PBIO Kraus operator $K_j=U_jP_j$ has exactly the standard PIO form: an incoherent unitary followed by a projector onto a subset of basis states. Thus every PBIO channel is a PIO channel in the rank one block case and this completes the proof.
    \end{proof}

The above Corollary shows that PBIO reduces to PIO when every block subspace is one dimensional. In this case, there is no internal structure within any block. Each block is spanned by a single basis vector and the intra block unitaries of PBIO can only multiply that vector by a phase. The extra generality of PBIO over PIO is therefore entirely due to multi dimensional block structure. Equivalently PBIO can perform operations genuinely beyond the reach of PIO. This corollary confirms that the standard coherence resource theory based on PIO is recovered as a special case of block coherence theory in trivial block limit.

\section{Block Dephasing Covariant Operation}
\label{sec:BDCO}
A quantum channel $\mathcal{E}:\mathcal{L(\mathcal{H})} \rightarrow \mathcal{L(\mathcal{H})}$ is called a block dephasing covariant operation (BDCO) if it is CPTP and satisfies 
$$\Delta \circ \mathcal{E}=\mathcal{E} \circ \Delta$$
For pure states $\ket{\psi},\ket{\phi} \in \mathcal{H}$ we set
$$\rho_\psi=\ket{\psi}\bra{\psi}, \qquad \rho_\phi=\ket{\phi}\bra{\phi}.$$
We denote their block components and components of block vectors by
\begin{equation*}
    \ket{\psi_\alpha}=\Pi_\alpha \ket{\psi}, \qquad p_\alpha=||\ket{\psi_\alpha}||^2=Tr(\Pi_\alpha \ket{\psi}\bra{\psi}),
\end{equation*}
\begin{equation*}
    \ket{\phi_\beta}=\Pi_\beta \ket{\phi}, \qquad q_\beta=||\ket{\phi_\beta}||^2=Tr(\Pi_\beta \ket{\phi}\bra{\phi})
\end{equation*}
The vectors $\tilde{p}=(p_\alpha)_\alpha$ and $\tilde{q}=(q_\beta)_\beta$ are block probability vectors.
We say that $\tilde{p}$ is majorized by $\tilde{q}$, written as $\tilde{p} \prec \tilde{q}$, if
$$\sum_{k=1}^{r} p_k^{\downarrow} \leq \sum_{k=1}^{r} q_k^{\downarrow}, \qquad r=1,2,...,m-1,$$
$$\sum_{k=1}^m p_k^{\downarrow}=\sum_{k=1}^m q_k^{\downarrow}$$ 

For each $\beta, \alpha$ whenever $q_\beta >0, p_\alpha >0$ we define the normalized vectors
\begin{equation}
    |\hat\phi_\beta\rangle=\frac{\ket{\phi_\beta}}{\sqrt{q_\beta}},
    \qquad
    |\hat\psi_\alpha\rangle=\frac{\ket{\psi_\alpha}}{\sqrt{p_\alpha}}
\label{eq 9}
\end{equation}
For the input and output pure states, $p_\alpha=0$ gives
$$\Pi_\alpha\ket{\psi}=\ket{\psi_\alpha}=0.$$
Therefore the $\alpha$-th input block has no contribution to the block decomposition of $\ket{\psi}$.\\
Similarly, if $q_\beta=0$, then
$$\Pi_\beta\ket{\phi}=\ket{\phi_\beta}=0,$$
hence $\beta$-th output block has no contribution to the block decomposition of $\ket{\phi}$.\\
Therefore, all expressions involving normalized block vectors will be understood only on $p_\alpha > 0$ and $q_\beta >0$. In other words, throughout our calculation in the next Theorem we will assume $\alpha \in R$ and $\beta \in Q$, where the set $R$ and $Q$ were defined in section \ref{sec:Preliminary}.
\begin{theorem}
  A given initial state $\ket{\psi}=\sum_\alpha \ket{\psi_\alpha}$ can be transformed deterministically to a given target state $\ket{\phi}=\sum_\beta \ket{\phi_\beta}$ by a BDCO channel if and only if the majorization relation
  $$\tilde{p} \prec \tilde{q}$$
  holds.
  \label{BDCO Theorem}
\end{theorem}

Since the output is a rank one pure state, every active Kraus branch must satisfy
$$K_l \ket{\psi}=c_l \ket{\phi}.$$
Next, we consider the blockwise image $\Pi_\beta K_l \ket{\psi_\alpha},$ namely the part of $K _l \ket{\psi_\alpha}$ lying in $\mathcal{H}_\beta.$ The BDCO condition then imposes orthogonality constraints among these blockwise images. Together with the rank one output condition, this implies that every nonzero blockwise image is parallel to the corresponding target block component $\ket{\phi_\beta}.$ These relations generate a doubly substochastic matrix $T$ satisfying
$$\tilde{p}=T \tilde{q}.$$
Hence $\tilde{p} \prec_w \tilde{q}$. Since both $\tilde{p}$ and $\tilde{q}$ are block probability vectors, weak majorization becomes ordinary majorization. Therefore,
$$\tilde{p} \prec \tilde{q}.$$
Conversely, we assume that
$$\tilde{p} \prec \tilde{q}.$$
By the majorization theorem, there exists a doubly stochastic matrix $T$ such that
$$\tilde{p}= T \tilde{q}.$$
Using the Birkhoff-von Neumann theorem, we construct Kraus operators corresponding to permutation branches. These Kraus operators send the active input block components of $\ket{\psi}$ to the appropriate output block components of $\ket{\phi},$ so that each active Kraus branch maps $\ket{\psi}$ to a scalar multiple of $\ket{\phi}.$ Additional Kraus operators are added only to complete trace preserving condition of the map $\mathcal{E}$. These auxiliary operators vanish on $\ket{\psi}$, so they do not affect the desired transformation. The constructed map is CPTP, satisfies the BDCO condition and also satisfy 
$$\mathcal{E}(\ket{\psi}\bra{\psi})=\ket{\phi}\bra{\phi}.$$
The detailed proof is given in Appendix \ref{Proof of BDCO Theorem}.

The above theorem shows that pure state transformation under BDCO is completely governed by the majorization relation between the block vectors of the input and output states. The relevant quantifiers are not the individual basis probabilities, but the block probability vectors of the input and output states,
$$p_\alpha=||\ket{\psi_\alpha}||^2=Tr(\Pi_\alpha \ket{\psi}\bra{\psi}),$$
$$q_\beta=||\ket{\phi_\beta}||^2=Tr(\Pi_\beta \ket{\phi}\bra{\phi}).$$
Thus BDCO does not use coherence inside each block. It only constrains how the total probability weight is distributed among the blocks.\\
The condition
$$\tilde{p} \prec \tilde{q},$$
means that the input block distribution $\tilde{p}$ is more spread out than the output block distribution $\tilde{q}$. Therefore, a BDCO can transform a pure state only into another pure state whose block weights are at least as concentrated as those of the input. Thus BDCO cannot create additional block coherence. It can only degrade or concentrate the block level superposition structure already present in the input.\\
This gives a direct block coherence analogue of the majorization criterion for pure state transformations under DIO. In the rank one case, each block contains only one basis vector and the block vector reduces to the usual diagonal probability vector. Hence above theorem reduces to the standard DIO pure state majorization condition. For higher rank blocks, the theorem is genuinely block structural.

\begin{corollary}
    For pure states $\ket{\psi}$ and $\ket{\phi} \in \mathcal{H}$, a transformation under SBIO is possible if and only if the block probability vector $\tilde{p}$ of $\ket{\psi}$ is majorized by the block probability vector $\tilde{q}$ of $\ket{\phi},\;\text{i.e.}\; \tilde{p} \prec \tilde{q}.$
\label{SBIO corollary}
\end{corollary}

The detailed verification is given in Appendix \ref{Appendix B}.\\

The above corollary shows that SBIO and BDCO have the same deterministic pure state convertibility power. Although SBIO is more restrictive at the Kraus operator level, the constructive proof of Theorem \ref{BDCO Theorem} shows that every BDCO allowed pure state transformation can also be implemented by an SBIO channel. Thus, for pure states, enlarging the class from SBIO to BDCO does not enlarge the set of reachable states.

\section{Maximally Block Coherent State as a Universal Resource}
\label{sec:Maximally block coherence state}
We now discuss an immediate application of the pure state transformation criterion. We choose a normalized vector 
$$\ket{\eta_\alpha} \in \mathcal{H_\alpha}, \qquad \langle \eta_\alpha | \eta_\alpha \rangle=1$$
from each block. We define the maximally block coherent pure state \cite{brub_second} as
$$\ket{\psi_m}= \frac{1}{\sqrt{m}} \sum_{\alpha=1}^m \ket{\eta_\alpha}.$$
For this state the block probability vector is 
$$\tilde{u}_m=\big( \frac{1}{m}, \frac{1}{m},...,\frac{1}{m}\big).$$

\begin{corollary}
    Let $\ket{\phi} \in \mathcal{H}$ be any pure state with block decomposition
    $$\ket{\phi}=\sum_{\beta =1}^m \ket{\phi_\beta}, \qquad \ket{\phi_\beta}= \Pi_\beta \ket{\phi},$$
    and 
    $$q_\beta=||\ket{\phi_\beta}||^2.$$
    Then the transformation from $\ket{\psi_m}$ to $\ket{\phi}$ is possible under BDCO and SBIO.
\end{corollary}

\begin{proof}
    The block probability vector of $\ket{\psi_m}$ is
    $$\tilde{u}_m=\big( \frac{1}{m}, \frac{1}{m},...,\frac{1}{m}\big).$$
    For any block probability vector $\tilde{q}=(q_1,q_2,...q_m)$, the uniform vector is majorized by $\tilde{q}$. Indeed for every $r=1,2,..,m$,
    $$\sum_{k=1}^r q_k^{\downarrow} \geq \frac{r}{m}.$$
    Therefore,
    $$\sum_{k=1}^r \frac{1}{m} \leq \sum_{k=1}^r q_k^{\downarrow},$$
    and hence 
    $$\tilde{u}_m \prec \tilde{q}.$$
    Therefore, the maximally block coherent state can be converted into any pure state under BDCO and SBIO. This completes the proof. 
\end{proof}

The above result ensures that the maximally block coherent state is the most resourceful pure state under the BDCO and SBIO.

\section{Geometry of block level state convertibility}
\label{sec:Numerical Examples}

Throughout this section, we denote the rank one probability vectors by
$$\tilde{p}=(p_0,p_1,p_2,p_3), \qquad \tilde{q}=(q_0,q_1,q_2,q_3).$$
For the block structure $1+2+1$, we define
$$B(\tilde{x})=(x_0, x_1+x_2, x_3).$$
Thus, $B(\tilde{p})$ and $B(\tilde{q})$ denote the corresponding block probability vectors. The symbol $\prec$ always denotes majorization after decreasing rearrangement.

\subsection{Input State Region for a Fixed Target State}

Figure \ref{fig:BDCO and DIO comparision} shows the input state region under DIO and BDCO for the fixed probability vector $q$ associated with a fixed target state, where
$$\tilde{q}=(0.45, 0.25, 0.20, 0.10).$$
Here we consider the block structure as $1+2+1$ and the corresponding projectors are $\Pi_1, \Pi_2$ and $\Pi_3$.\\
Where
$$\Pi_1=\ket{0}\bra{0}, \;\; \Pi_2=\ket{1}\bra{1}+\ket{2}\bra{2},\;\;\Pi_3=\ket{3}\bra{3}.$$
The four-dimensional probability simplex is represented as a tetrahedron, where each point corresponds to an input probability vector 
$$\tilde{p}=(p_0, p_1, p_2, p_3).$$
\begin{figure}[htbp]
\centering
\includegraphics[width=0.45\textwidth]{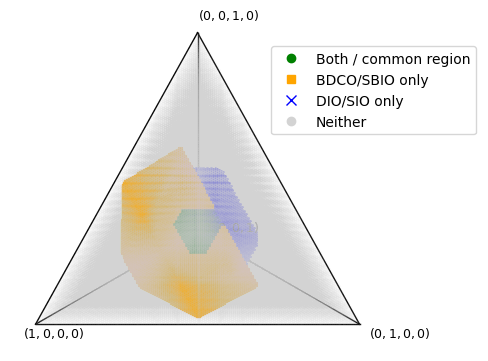}
\caption{Input state region for a fixed target probability vector $\tilde{q}$}
\label{fig:BDCO and DIO comparision}
\end{figure}
For DIO, the conversion condition is $\tilde{p} \prec \tilde{q}$, 
and for BDCO the conversion condition becomes $B(\tilde{p}) \prec B(\tilde{q})$. Here
$$B(\tilde{q})=(0.45, 0.45, 0.10).$$
The green region represents input states satisfying both criteria, the orange region represents states satisfying only the BDCO block majorization criterion, the blue region represents states satisfying only the ordinary DIO criterion, and the gray region represents states satisfying neither criterion.

\subsection{Achievable Target State Region from a Fixed Input State}

Figure \ref{fig:BDCO and DIO comparision 2} shows the achievable target state region from the fixed input probability vector
$$\tilde{p}=(0.50, 0.25, 0.15, 0.10)$$
under DIO and BDCO for the block structure $1+2+1$.
\begin{figure}[ht]
\centering
\includegraphics[width=0.51\textwidth]{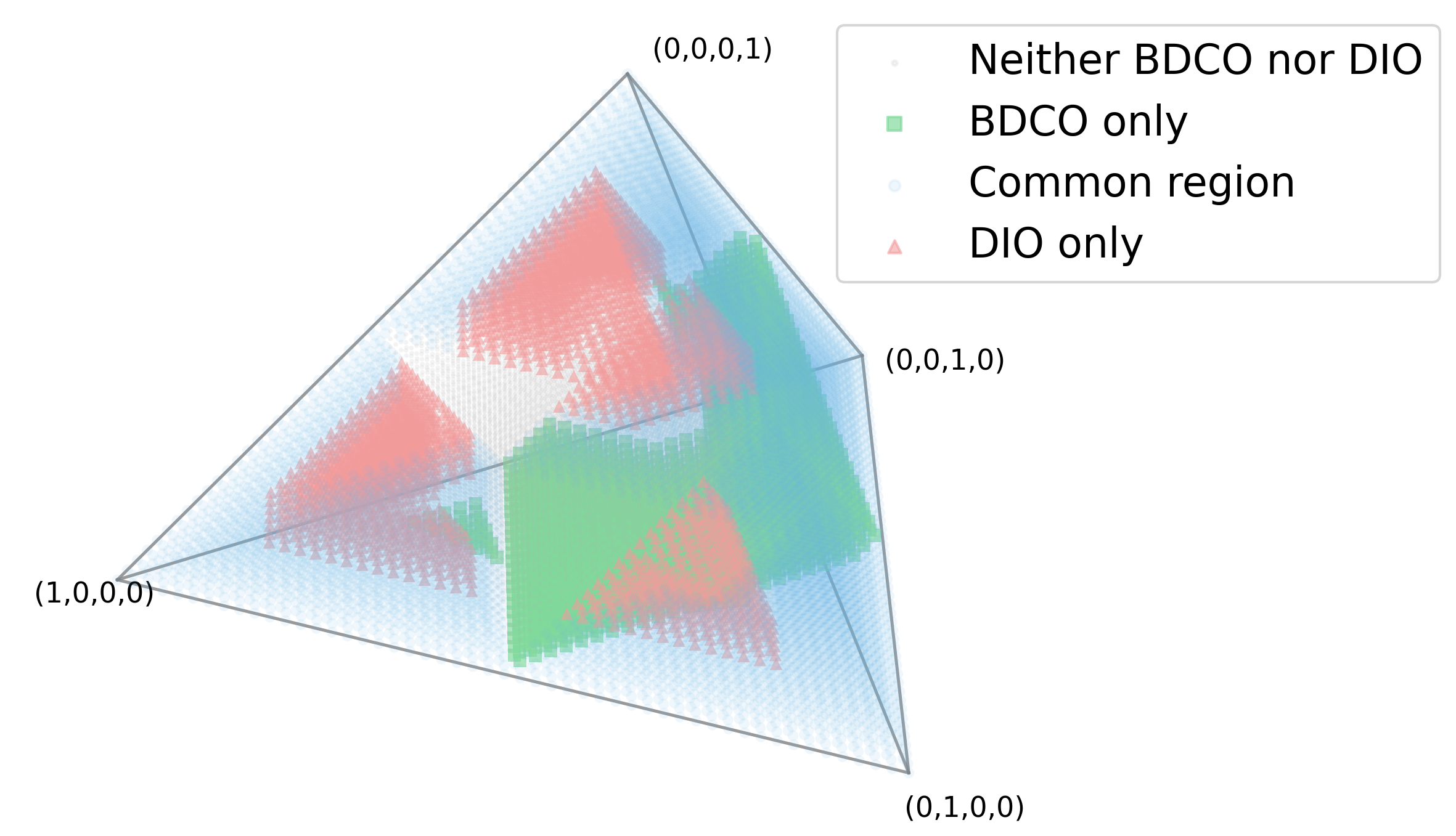}
\caption{Achievable target state region for a fixed input probability vector $\tilde{p}$}
\label{fig:BDCO and DIO comparision 2}
\end{figure}
The four-dimensional probability simplex is represented as a tetrahedron, where each point corresponds to a target probability vector 
$$\tilde{q}=(q_0, q_1, q_2, q_3).$$
For DIO the conversion condition is 
$$\tilde{p} \prec \tilde{q},$$ 
and under BDCO the conversion condition becomes
$$B(\tilde{p}) \prec B(\tilde{q}),$$
for the chosen input,
$$B(\tilde{p})=(0.50, 0.40, 0.10).$$
The blue region represents targets achievable under both criteria, the pink region represents targets achievable only under DIO, the green region represents targets achievable only under BDCO and the gray region represents targets achievable under neither criterion.

\subsection{Comparison of the Reversible region under DIO and BDCO}

\begin{figure}[htbp]
\centering
\includegraphics[width=0.45\textwidth]{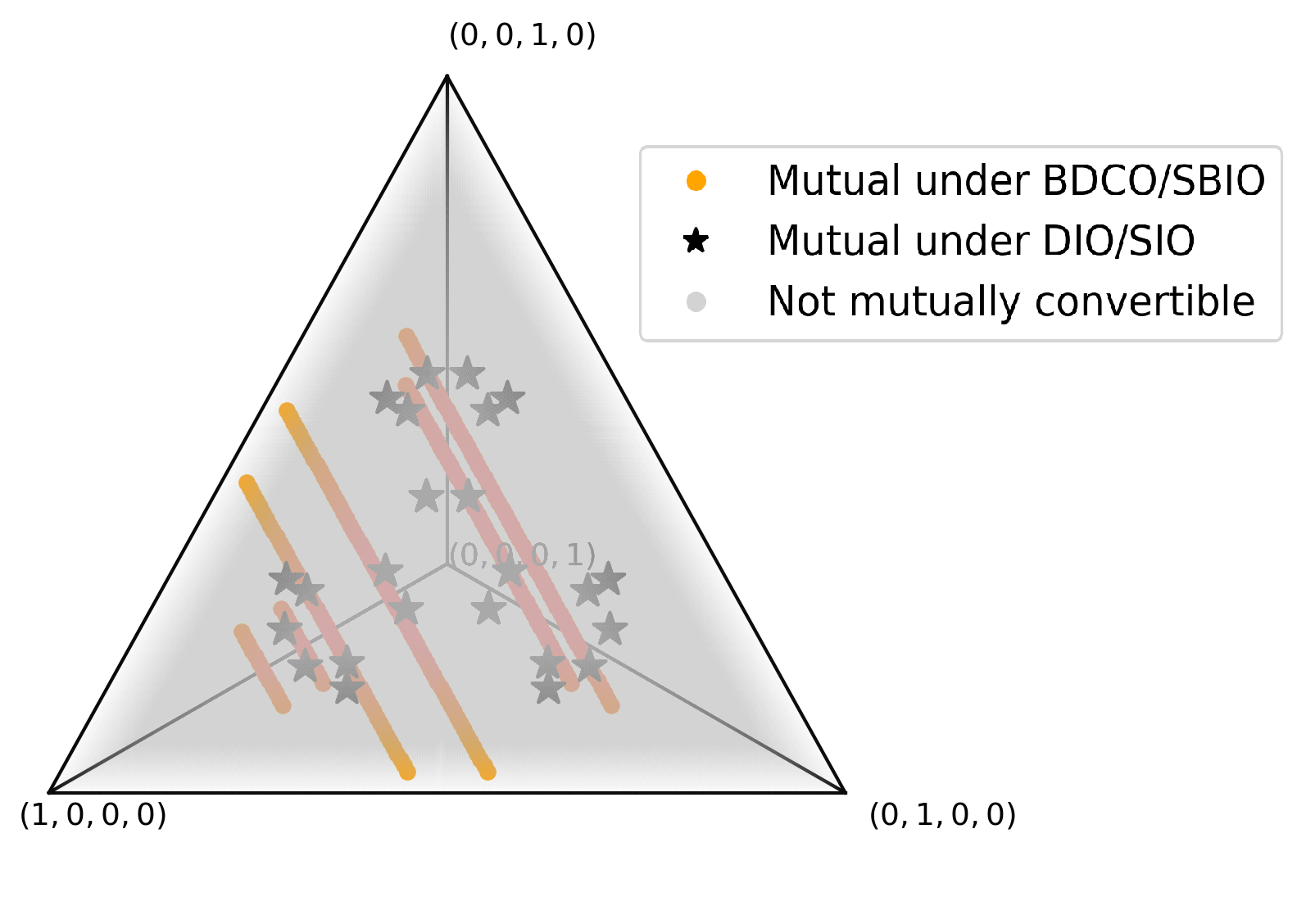}
\caption{Reversible region for a fixed input probability vector $\tilde{p}$}
\label{fig:BDCO and DIO comparision 3}
\end{figure}

Figure \ref{fig:BDCO and DIO comparision 3} shows the reversible region for the fixed probability vector 
$$\tilde{p}=(0.50, 0.25, 0.15, 0.10)$$
under DIO and BDCO for the block structure $1+2+1$.\\
The four-dimensional probability simplex is represented as a tetrahedron, where each point denotes a probability vector $\tilde{q}=(q_0, q_1, q_2, q_3).$
For DIO, mutual convertibility requires
$$\tilde{p} \prec \tilde{q} \;\;\text{and}\;\; \tilde{q} \prec \tilde{p},$$
which is equivalent to 
$$\tilde{p}^\downarrow = \tilde{q}^\downarrow.$$
For BDCO, mutual convertibility is governed by the block probability vectors 
$$B(\tilde{p})=(p_0, p_1+p_2, p_3), \qquad B(\tilde{q})=(q_0, q_1+q_2, q_3),$$
and requires 
$$B(\tilde{p}) \prec B(\tilde{q}) \;\;\text{and}\;\; B(\tilde{q}) \prec B(\tilde{p}),$$
equivalently
$$B(\tilde{p})^\downarrow = B(\tilde{q})^\downarrow.$$
For the chosen $\tilde{p}$, the BDCO reversible set is determined by
$$B(\tilde{q})^\downarrow=(0.50, 0.40, 0.10),$$
whereas the DIO reversible set is determined by
$$\tilde{q}^\downarrow=(0.50, 0.25, 0.15, 0.10).$$
The orange points denote states mutually convertible with $\tilde{p}$ under BDCO, while the black starred points denote states mutually convertible with $\tilde{p}$ under DIO. 

\subsection{Physical Interpretation of Block Level Convertibility Regions}
The three geometrical plots provide a geometric interpretation of the difference between the rank-one coherence theory and the block coherence theory. In the case of DIO, the relevant probability vector is 
$$\tilde{p}=(p_0, p_1, p_2, p_3),$$
and the deterministic pure state transformation is governed by the ordinary majorization condition
$$\tilde{p} \prec \tilde{q}.$$
In contrast, under BDCO with block structure $1+2+1$, the relevant object is the block probability vector
$$B(\tilde{p})=(p_0, p_1+p_2, p_3).$$
Thus, the BDCO transformation condition becomes
$$B(\tilde{p}) \prec B(\tilde{q}).$$
This difference reflects the physical fact that BDCO does not distinguish the individual probabilities $p_1$ and $p_2$ inside the rank two block. The Only operationally relevant component inside the rank two block is $p_1+p_2.$\\
Overall, these three plots show that BDCO and DIO describe state transformations in different ways. DIO looks at the population of every individual basis state. In contrast, BDCO looks only at the total population inside each block. Therefore, two states that look different under DIO may look the same under BDCO if they have the same block populations.\\
This means that BDCO is not just a direct extension of DIO from basis states to blocks. The above figures show that the BDCO and DIO regions may overlap, but it is not the case that one region is simply contained inside the other.\\
The importance of BDCO is that it gives a natural description when the physically relevant coherence is between blocks, not between individual basis states. In such cases, changes inside a block are not treated as resourceful, while coherence between different blocks is the main resource. Thus, BDCO provides a useful and physically meaningful way to study pure transformation at the level of block structure.\\
The mutual convertibility plot is especially important because it shows that BDCO can give a larger and physically meaningful region of mutual transformations. Some transformations that are forbidden under DIO become possible under BDCO.

\section{Conclusion}
\label{sec:Conclusion}
This work provides a complete structural description of deterministic pure state transformations in block coherence theory for PBIO, SBIO, and BDCO. For PBIO, the main result is that, under a physically meaningful nondegeneracy condition, any deterministic pure state conversion is effectively trivial at the resource level. It must be realized by a single active Kraus operator and hence by a block incoherent unitary. When the nondegeneracy assumption is removed, the condition is still highly constrained. It requires each active branch to match the target structure up to a common branch-dependent scalar.

For SBIO and BDCO, the convertibility condition is found to be exactly the condition of majorization of block probability vectors. This establishes a clean ordering of pure states based on their block probability vectors, and the converse part supplies an explicit Kraus realization of the allowed transformation. It shows that these operations can not create additional block coherence. They can only rearrange the distribution of weight among blocks. The coincidence of the SBIO and BDCO criteria is particularly significant, since it shows that enlarging the operation class from strict block dephasing covariance to full block dephasing covariance does not increase the set of reachable pure states.

The state transformation conditions align smoothly with the standard coherence framework when all the projectors are rank one. Thus, the present results genuinely extend known pure state coherence convertibility conditions to a more general setting. At the same time, the block setting reveals a broader resource structure. Under a fixed orthogonal decomposition of the Hilbert space $\mathcal{H}=\bigoplus_{\alpha=1}^{m}\mathcal{H}_\alpha$, the pure state with uniform block probability vector $(\frac{1}{m},\frac{1}{m},...,\frac{1}{m})$ emerges as the maximally block coherence state under SBIO and BDCO. The geometric plots confirm that nontrivial block partitions can enable deterministic transformations forbidden in the standard coherence theory. The mutual convertibility plot is especially useful because it identifies the reversible sector of the transformation. BDCO can group together states that are distinct under DIO but equivalent at the block level.

Overall, our work shows that block coherence supports a physically meaningful convertibility theory, with majorization serving as the central ordering principle for pure states under SBIO and BDCO. A logical next step is to extend these results to mixed states, while also exploring various types of state transformations and examining block coherence quantifiers. 

\section*{ACKNOWLEDGMENT}
Priyabrata Char acknowledges the financial support of the Institute Postdoctoral Fellowship, IIT Guwahati, India. The authors D. Sarkar and I. Chattopadhyay acknowledges DST-FIST India.


\begin{thebibliography}{27}%
\makeatletter
\providecommand \@ifxundefined [1]{%
 \@ifx{#1\undefined}
}%
\providecommand \@ifnum [1]{%
 \ifnum #1\expandafter \@firstoftwo
 \else \expandafter \@secondoftwo
 \fi
}%
\providecommand \@ifx [1]{%
 \ifx #1\expandafter \@firstoftwo
 \else \expandafter \@secondoftwo
 \fi
}%
\providecommand \natexlab [1]{#1}%
\providecommand \enquote  [1]{``#1''}%
\providecommand \bibnamefont  [1]{#1}%
\providecommand \bibfnamefont [1]{#1}%
\providecommand \citenamefont [1]{#1}%
\providecommand \href@noop [0]{\@secondoftwo}%
\providecommand \href [0]{\begingroup \@sanitize@url \@href}%
\providecommand \@href[1]{\@@startlink{#1}\@@href}%
\providecommand \@@href[1]{\endgroup#1\@@endlink}%
\providecommand \@sanitize@url [0]{\catcode `\\12\catcode `\$12\catcode `\&12\catcode `\#12\catcode `\^12\catcode `\_12\catcode `\%12\relax}%
\providecommand \@@startlink[1]{}%
\providecommand \@@endlink[0]{}%
\providecommand \url  [0]{\begingroup\@sanitize@url \@url }%
\providecommand \@url [1]{\endgroup\@href {#1}{\urlprefix }}%
\providecommand \urlprefix  [0]{URL }%
\providecommand \Eprint [0]{\href }%
\providecommand \doibase [0]{http://dx.doi.org/}%
\providecommand \selectlanguage [0]{\@gobble}%
\providecommand \bibinfo  [0]{\@secondoftwo}%
\providecommand \bibfield  [0]{\@secondoftwo}%
\providecommand \translation [1]{[#1]}%
\providecommand \BibitemOpen [0]{}%
\providecommand \bibitemStop [0]{}%
\providecommand \bibitemNoStop [0]{.\EOS\space}%
\providecommand \EOS [0]{\spacefactor3000\relax}%
\providecommand \BibitemShut  [1]{\csname bibitem#1\endcsname}%
\let\auto@bib@innerbib\@empty
\bibitem [{\citenamefont {Baumgratz}\ \emph {et~al.}(2014)\citenamefont {Baumgratz}, \citenamefont {Cramer},\ and\ \citenamefont {Plenio}}]{Baumgratz_2014_Quantifying}%
  \BibitemOpen
  \bibfield  {author} {\bibinfo {author} {\bibfnamefont {T.}~\bibnamefont {Baumgratz}}, \bibinfo {author} {\bibfnamefont {M.}~\bibnamefont {Cramer}}, \ and\ \bibinfo {author} {\bibfnamefont {M.~B.}\ \bibnamefont {Plenio}},\ }\href {\doibase 10.1103/PhysRevLett.113.140401} {\bibfield  {journal} {\bibinfo  {journal} {Physical Review Letters}\ }\textbf {\bibinfo {volume} {113}},\ \bibinfo {pages} {140401} (\bibinfo {year} {2014})}\BibitemShut {NoStop}%
\bibitem [{\citenamefont {Streltsov}\ \emph {et~al.}(2017)\citenamefont {Streltsov}, \citenamefont {Adesso},\ and\ \citenamefont {Plenio}}]{Streltsov_2017_Colloquium}%
  \BibitemOpen
  \bibfield  {author} {\bibinfo {author} {\bibfnamefont {A.}~\bibnamefont {Streltsov}}, \bibinfo {author} {\bibfnamefont {G.}~\bibnamefont {Adesso}}, \ and\ \bibinfo {author} {\bibfnamefont {M.~B.}\ \bibnamefont {Plenio}},\ }\href {\doibase 10.1103/RevModPhys.89.041003} {\bibfield  {journal} {\bibinfo  {journal} {Reviews of Modern Physics}\ }\textbf {\bibinfo {volume} {89}},\ \bibinfo {pages} {041003} (\bibinfo {year} {2017})}\BibitemShut {NoStop}%
\bibitem [{\citenamefont {Giovannetti}\ \emph {et~al.}(2011)\citenamefont {Giovannetti}, \citenamefont {Lloyd},\ and\ \citenamefont {Maccone}}]{Giovannetti_2011_AdvancesMetrology}%
  \BibitemOpen
  \bibfield  {author} {\bibinfo {author} {\bibfnamefont {V.}~\bibnamefont {Giovannetti}}, \bibinfo {author} {\bibfnamefont {S.}~\bibnamefont {Lloyd}}, \ and\ \bibinfo {author} {\bibfnamefont {L.}~\bibnamefont {Maccone}},\ }\href {\doibase 10.1038/nphoton.2011.35} {\bibfield  {journal} {\bibinfo  {journal} {Nature Photonics}\ }\textbf {\bibinfo {volume} {5}},\ \bibinfo {pages} {222} (\bibinfo {year} {2011})}\BibitemShut {NoStop}%
\bibitem [{\citenamefont {Hillery}(2016)}]{Hillery_2016_DeutschJozsa}%
  \BibitemOpen
  \bibfield  {author} {\bibinfo {author} {\bibfnamefont {M.}~\bibnamefont {Hillery}},\ }\href {\doibase 10.1103/PhysRevA.93.012111} {\bibfield  {journal} {\bibinfo  {journal} {Physical Review A}\ }\textbf {\bibinfo {volume} {93}},\ \bibinfo {pages} {012111} (\bibinfo {year} {2016})}\BibitemShut {NoStop}%
\bibitem [{\citenamefont {Lostaglio}\ \emph {et~al.}(2015{\natexlab{a}})\citenamefont {Lostaglio}, \citenamefont {Jennings},\ and\ \citenamefont {Rudolph}}]{Lostaglio_2015_Description}%
  \BibitemOpen
  \bibfield  {author} {\bibinfo {author} {\bibfnamefont {M.}~\bibnamefont {Lostaglio}}, \bibinfo {author} {\bibfnamefont {D.}~\bibnamefont {Jennings}}, \ and\ \bibinfo {author} {\bibfnamefont {T.}~\bibnamefont {Rudolph}},\ }\href {\doibase 10.1038/ncomms7383} {\bibfield  {journal} {\bibinfo  {journal} {Nature Communications}\ }\textbf {\bibinfo {volume} {6}},\ \bibinfo {pages} {6383} (\bibinfo {year} {2015}{\natexlab{a}})}\BibitemShut {NoStop}%
\bibitem [{\citenamefont {Lostaglio}\ \emph {et~al.}(2015{\natexlab{b}})\citenamefont {Lostaglio}, \citenamefont {Korzekwa}, \citenamefont {Jennings},\ and\ \citenamefont {Rudolph}}]{Lostaglio_2015_TimeTranslation}%
  \BibitemOpen
  \bibfield  {author} {\bibinfo {author} {\bibfnamefont {M.}~\bibnamefont {Lostaglio}}, \bibinfo {author} {\bibfnamefont {K.}~\bibnamefont {Korzekwa}}, \bibinfo {author} {\bibfnamefont {D.}~\bibnamefont {Jennings}}, \ and\ \bibinfo {author} {\bibfnamefont {T.}~\bibnamefont {Rudolph}},\ }\href {\doibase 10.1103/PhysRevX.5.021001} {\bibfield  {journal} {\bibinfo  {journal} {Physical Review X}\ }\textbf {\bibinfo {volume} {5}},\ \bibinfo {pages} {021001} (\bibinfo {year} {2015}{\natexlab{b}})}\BibitemShut {NoStop}%
\bibitem [{\citenamefont {Narasimhachar}\ and\ \citenamefont {Gour}(2015)}]{Narasimhachar_2015_LowTemp}%
  \BibitemOpen
  \bibfield  {author} {\bibinfo {author} {\bibfnamefont {V.}~\bibnamefont {Narasimhachar}}\ and\ \bibinfo {author} {\bibfnamefont {G.}~\bibnamefont {Gour}},\ }\href {\doibase 10.1038/ncomms8689} {\bibfield  {journal} {\bibinfo  {journal} {Nature Communications}\ }\textbf {\bibinfo {volume} {6}},\ \bibinfo {pages} {7689} (\bibinfo {year} {2015})}\BibitemShut {NoStop}%
\bibitem [{\citenamefont {Huelga}\ and\ \citenamefont {Plenio}(2013)}]{Huelga_2013_Vibrations}%
  \BibitemOpen
  \bibfield  {author} {\bibinfo {author} {\bibfnamefont {S.~F.}\ \bibnamefont {Huelga}}\ and\ \bibinfo {author} {\bibfnamefont {M.~B.}\ \bibnamefont {Plenio}},\ }\href {\doibase 10.1080/00405000.2013.829687} {\bibfield  {journal} {\bibinfo  {journal} {Contemporary Physics}\ }\textbf {\bibinfo {volume} {54}},\ \bibinfo {pages} {181} (\bibinfo {year} {2013})}\BibitemShut {NoStop}%
\bibitem [{\citenamefont {Chitambar}\ and\ \citenamefont {Gour}(2016)}]{Chitambar_2016_Critical}%
  \BibitemOpen
  \bibfield  {author} {\bibinfo {author} {\bibfnamefont {E.}~\bibnamefont {Chitambar}}\ and\ \bibinfo {author} {\bibfnamefont {G.}~\bibnamefont {Gour}},\ }\href {\doibase 10.1103/PhysRevLett.117.030401} {\bibfield  {journal} {\bibinfo  {journal} {Physical Review Letters}\ }\textbf {\bibinfo {volume} {117}},\ \bibinfo {pages} {030401} (\bibinfo {year} {2016})}\BibitemShut {NoStop}%
\bibitem [{\citenamefont {Yadin}\ \emph {et~al.}(2016)\citenamefont {Yadin}, \citenamefont {Ma}, \citenamefont {Girolami}, \citenamefont {Gu},\ and\ \citenamefont {Vedral}}]{Yadin_2016_Processes}%
  \BibitemOpen
  \bibfield  {author} {\bibinfo {author} {\bibfnamefont {B.}~\bibnamefont {Yadin}}, \bibinfo {author} {\bibfnamefont {J.}~\bibnamefont {Ma}}, \bibinfo {author} {\bibfnamefont {D.}~\bibnamefont {Girolami}}, \bibinfo {author} {\bibfnamefont {M.}~\bibnamefont {Gu}}, \ and\ \bibinfo {author} {\bibfnamefont {V.}~\bibnamefont {Vedral}},\ }\href {\doibase 10.1103/PhysRevX.6.041028} {\bibfield  {journal} {\bibinfo  {journal} {Physical Review X}\ }\textbf {\bibinfo {volume} {6}},\ \bibinfo {pages} {041028} (\bibinfo {year} {2016})}\BibitemShut {NoStop}%
\bibitem [{\citenamefont {Winter}\ and\ \citenamefont {Yang}(2016)}]{Winter_2016_Operational}%
  \BibitemOpen
  \bibfield  {author} {\bibinfo {author} {\bibfnamefont {A.}~\bibnamefont {Winter}}\ and\ \bibinfo {author} {\bibfnamefont {D.}~\bibnamefont {Yang}},\ }\href {\doibase 10.1103/PhysRevLett.116.120404} {\bibfield  {journal} {\bibinfo  {journal} {Physical Review Letters}\ }\textbf {\bibinfo {volume} {116}},\ \bibinfo {pages} {120404} (\bibinfo {year} {2016})}\BibitemShut {NoStop}%
\bibitem [{\citenamefont {Nielsen}(1999)}]{Nielsen_1999_Conditions}%
  \BibitemOpen
  \bibfield  {author} {\bibinfo {author} {\bibfnamefont {M.~A.}\ \bibnamefont {Nielsen}},\ }\href {\doibase 10.1103/PhysRevLett.83.436} {\bibfield  {journal} {\bibinfo  {journal} {Physical Review Letters}\ }\textbf {\bibinfo {volume} {83}},\ \bibinfo {pages} {436} (\bibinfo {year} {1999})}\BibitemShut {NoStop}%
\bibitem [{\citenamefont {Du}\ \emph {et~al.}(2015)\citenamefont {Du}, \citenamefont {Bai},\ and\ \citenamefont {Guo}}]{Du_2015_Conditions}%
  \BibitemOpen
  \bibfield  {author} {\bibinfo {author} {\bibfnamefont {S.}~\bibnamefont {Du}}, \bibinfo {author} {\bibfnamefont {Z.}~\bibnamefont {Bai}}, \ and\ \bibinfo {author} {\bibfnamefont {Y.}~\bibnamefont {Guo}},\ }\href {\doibase 10.1103/PhysRevA.91.052120} {\bibfield  {journal} {\bibinfo  {journal} {Physical Review A}\ }\textbf {\bibinfo {volume} {91}},\ \bibinfo {pages} {052120} (\bibinfo {year} {2015})}\BibitemShut {NoStop}%
\bibitem [{\citenamefont {Marvian}\ and\ \citenamefont {Spekkens}(2016)}]{Marvian_2016_Speakable}%
  \BibitemOpen
  \bibfield  {author} {\bibinfo {author} {\bibfnamefont {I.}~\bibnamefont {Marvian}}\ and\ \bibinfo {author} {\bibfnamefont {R.~W.}\ \bibnamefont {Spekkens}},\ }\href {\doibase 10.1103/PhysRevA.94.052324} {\bibfield  {journal} {\bibinfo  {journal} {Physical Review A}\ }\textbf {\bibinfo {volume} {94}},\ \bibinfo {pages} {052324} (\bibinfo {year} {2016})}\BibitemShut {NoStop}%
\bibitem [{\citenamefont {Theurer}\ \emph {et~al.}(2017)\citenamefont {Theurer}, \citenamefont {Killoran}, \citenamefont {Egloff},\ and\ \citenamefont {Plenio}}]{Theurer_2017_Superposition}%
  \BibitemOpen
  \bibfield  {author} {\bibinfo {author} {\bibfnamefont {T.}~\bibnamefont {Theurer}}, \bibinfo {author} {\bibfnamefont {N.}~\bibnamefont {Killoran}}, \bibinfo {author} {\bibfnamefont {D.}~\bibnamefont {Egloff}}, \ and\ \bibinfo {author} {\bibfnamefont {M.~B.}\ \bibnamefont {Plenio}},\ }\href {\doibase 10.1103/PhysRevLett.119.230401} {\bibfield  {journal} {\bibinfo  {journal} {Physical Review Letters}\ }\textbf {\bibinfo {volume} {119}},\ \bibinfo {pages} {230401} (\bibinfo {year} {2017})}\BibitemShut {NoStop}%
\bibitem [{\citenamefont {Marvian}\ and\ \citenamefont {Spekkens}(2014)}]{Marvian_2014_Noether}%
  \BibitemOpen
  \bibfield  {author} {\bibinfo {author} {\bibfnamefont {I.}~\bibnamefont {Marvian}}\ and\ \bibinfo {author} {\bibfnamefont {R.~W.}\ \bibnamefont {Spekkens}},\ }\href {\doibase 10.1038/ncomms4821} {\bibfield  {journal} {\bibinfo  {journal} {Nature Communications}\ }\textbf {\bibinfo {volume} {5}},\ \bibinfo {pages} {3821} (\bibinfo {year} {2014})}\BibitemShut {NoStop}%
\bibitem [{\citenamefont {Mani}\ \emph {et~al.}(2024)\citenamefont {Mani}, \citenamefont {Rezazadeh},\ and\ \citenamefont {Karimipour}}]{Coherence_subspaces}%
  \BibitemOpen
  \bibfield  {author} {\bibinfo {author} {\bibfnamefont {A.}~\bibnamefont {Mani}}, \bibinfo {author} {\bibfnamefont {F.}~\bibnamefont {Rezazadeh}}, \ and\ \bibinfo {author} {\bibfnamefont {V.}~\bibnamefont {Karimipour}},\ }\href {\doibase 10.1103/PhysRevA.109.012435} {\bibfield  {journal} {\bibinfo  {journal} {Phys. Rev. A}\ }\textbf {\bibinfo {volume} {109}},\ \bibinfo {pages} {012435} (\bibinfo {year} {2024})}\BibitemShut {NoStop}%
\bibitem [{\citenamefont {Bischof}\ \emph {et~al.}(2019{\natexlab{a}})\citenamefont {Bischof}, \citenamefont {Kampermann},\ and\ \citenamefont {Bru\ss{}}}]{block_coherence}%
  \BibitemOpen
  \bibfield  {author} {\bibinfo {author} {\bibfnamefont {F.}~\bibnamefont {Bischof}}, \bibinfo {author} {\bibfnamefont {H.}~\bibnamefont {Kampermann}}, \ and\ \bibinfo {author} {\bibfnamefont {D.}~\bibnamefont {Bru\ss{}}},\ }\href {\doibase 10.1103/PhysRevLett.123.110402} {\bibfield  {journal} {\bibinfo  {journal} {Phys. Rev. Lett.}\ }\textbf {\bibinfo {volume} {123}},\ \bibinfo {pages} {110402} (\bibinfo {year} {2019}{\natexlab{a}})}\BibitemShut {NoStop}%
\bibitem [{\citenamefont {Dey}\ \emph {et~al.}(2024)\citenamefont {Dey}, \citenamefont {Chakraborty}, \citenamefont {Char}, \citenamefont {Chattopadhyay}, \citenamefont {Bhar},\ and\ \citenamefont {Sarkar}}]{Dey:2019rje}%
  \BibitemOpen
  \bibfield  {author} {\bibinfo {author} {\bibfnamefont {P.~K.}\ \bibnamefont {Dey}}, \bibinfo {author} {\bibfnamefont {D.}~\bibnamefont {Chakraborty}}, \bibinfo {author} {\bibfnamefont {P.}~\bibnamefont {Char}}, \bibinfo {author} {\bibfnamefont {I.}~\bibnamefont {Chattopadhyay}}, \bibinfo {author} {\bibfnamefont {A.}~\bibnamefont {Bhar}}, \ and\ \bibinfo {author} {\bibfnamefont {D.}~\bibnamefont {Sarkar}},\ }\href {\doibase 10.26421/qic24.9-10-2} {\bibfield  {journal} {\bibinfo  {journal} {Quant. Inf. Comput.}\ }\textbf {\bibinfo {volume} {24}},\ \bibinfo {pages} {734} (\bibinfo {year} {2024})}\BibitemShut {NoStop}%
\bibitem [{\citenamefont {Xu}\ \emph {et~al.}(2020)\citenamefont {Xu}, \citenamefont {Shao},\ and\ \citenamefont {Fei}}]{brub_second}%
  \BibitemOpen
  \bibfield  {author} {\bibinfo {author} {\bibfnamefont {J.}~\bibnamefont {Xu}}, \bibinfo {author} {\bibfnamefont {L.-H.}\ \bibnamefont {Shao}}, \ and\ \bibinfo {author} {\bibfnamefont {S.-M.}\ \bibnamefont {Fei}},\ }\href {\doibase 10.1103/PhysRevA.102.012411} {\bibfield  {journal} {\bibinfo  {journal} {Phys. Rev. A}\ }\textbf {\bibinfo {volume} {102}},\ \bibinfo {pages} {012411} (\bibinfo {year} {2020})}\BibitemShut {NoStop}%
\bibitem [{\citenamefont {Ren}\ \emph {et~al.}(2021)\citenamefont {Ren}, \citenamefont {Gao}, \citenamefont {Ren}, \citenamefont {Wang},\ and\ \citenamefont {Bai}}]{Ren:2021deq}%
  \BibitemOpen
  \bibfield  {author} {\bibinfo {author} {\bibfnamefont {L.-H.}\ \bibnamefont {Ren}}, \bibinfo {author} {\bibfnamefont {M.}~\bibnamefont {Gao}}, \bibinfo {author} {\bibfnamefont {J.}~\bibnamefont {Ren}}, \bibinfo {author} {\bibfnamefont {Z.~D.}\ \bibnamefont {Wang}}, \ and\ \bibinfo {author} {\bibfnamefont {Y.-K.}\ \bibnamefont {Bai}},\ }\href {\doibase 10.1088/1367-2630/abd9e6} {\bibfield  {journal} {\bibinfo  {journal} {New J. Phys.}\ }\textbf {\bibinfo {volume} {23}},\ \bibinfo {pages} {043053} (\bibinfo {year} {2021})}\BibitemShut {NoStop}%
\bibitem [{\citenamefont {Kim}\ \emph {et~al.}(2021)\citenamefont {Kim}, \citenamefont {Xiong}, \citenamefont {Kumar},\ and\ \citenamefont {Wu}}]{kim2021converting}%
  \BibitemOpen
  \bibfield  {author} {\bibinfo {author} {\bibfnamefont {S.}~\bibnamefont {Kim}}, \bibinfo {author} {\bibfnamefont {C.}~\bibnamefont {Xiong}}, \bibinfo {author} {\bibfnamefont {A.}~\bibnamefont {Kumar}}, \ and\ \bibinfo {author} {\bibfnamefont {J.}~\bibnamefont {Wu}},\ }\href@noop {} {\bibfield  {journal} {\bibinfo  {journal} {Physical Review A}\ }\textbf {\bibinfo {volume} {103}},\ \bibinfo {pages} {052418} (\bibinfo {year} {2021})}\BibitemShut {NoStop}%
\bibitem [{\citenamefont {Bischof}\ \emph {et~al.}(2019{\natexlab{b}})\citenamefont {Bischof}, \citenamefont {Kampermann},\ and\ \citenamefont {Bru{\ss}}}]{Bischof_2019_POVMCoherence}%
  \BibitemOpen
  \bibfield  {author} {\bibinfo {author} {\bibfnamefont {F.}~\bibnamefont {Bischof}}, \bibinfo {author} {\bibfnamefont {H.}~\bibnamefont {Kampermann}}, \ and\ \bibinfo {author} {\bibfnamefont {D.}~\bibnamefont {Bru{\ss}}},\ }\href {\doibase 10.1103/PhysRevLett.123.110402} {\bibfield  {journal} {\bibinfo  {journal} {Physical Review Letters}\ }\textbf {\bibinfo {volume} {123}},\ \bibinfo {pages} {110402} (\bibinfo {year} {2019}{\natexlab{b}})}\BibitemShut {NoStop}%
\bibitem [{\citenamefont {Bischof}\ \emph {et~al.}(2021)\citenamefont {Bischof}, \citenamefont {Kampermann},\ and\ \citenamefont {Bru{\ss}}}]{Bischof_2021_Quantifying}%
  \BibitemOpen
  \bibfield  {author} {\bibinfo {author} {\bibfnamefont {F.}~\bibnamefont {Bischof}}, \bibinfo {author} {\bibfnamefont {H.}~\bibnamefont {Kampermann}}, \ and\ \bibinfo {author} {\bibfnamefont {D.}~\bibnamefont {Bru{\ss}}},\ }\href {\doibase 10.1103/PhysRevA.103.032429} {\bibfield  {journal} {\bibinfo  {journal} {Physical Review A}\ }\textbf {\bibinfo {volume} {103}},\ \bibinfo {pages} {032429} (\bibinfo {year} {2021})}\BibitemShut {NoStop}%
\bibitem [{\citenamefont {Marshall}\ \emph {et~al.}(2011)\citenamefont {Marshall}, \citenamefont {Olkin},\ and\ \citenamefont {Arnold}}]{Marshall_2011_Majorization}%
  \BibitemOpen
  \bibfield  {author} {\bibinfo {author} {\bibfnamefont {A.~W.}\ \bibnamefont {Marshall}}, \bibinfo {author} {\bibfnamefont {I.}~\bibnamefont {Olkin}}, \ and\ \bibinfo {author} {\bibfnamefont {B.~C.}\ \bibnamefont {Arnold}},\ }\href {\doibase 10.1007/978-0-387-68276-1} {\emph {\bibinfo {title} {{Inequalities: Theory of Majorization and Its Applications}}}},\ \bibinfo {edition} {2nd}\ ed.,\ Springer Series in Statistics\ (\bibinfo  {publisher} {Springer},\ \bibinfo {address} {New York},\ \bibinfo {year} {2011})\BibitemShut {NoStop}%
\bibitem [{\citenamefont {Bhatia}(1997)}]{Bhatia_1997_MatrixAnalysis}%
  \BibitemOpen
  \bibfield  {author} {\bibinfo {author} {\bibfnamefont {R.}~\bibnamefont {Bhatia}},\ }\href {\doibase 10.1007/978-1-4612-0653-8} {\emph {\bibinfo {title} {{Matrix Analysis}}}},\ \bibinfo {series} {Graduate Texts in Mathematics}, Vol.\ \bibinfo {volume} {169}\ (\bibinfo  {publisher} {Springer},\ \bibinfo {address} {New York},\ \bibinfo {year} {1997})\BibitemShut {NoStop}%
\bibitem [{\citenamefont {Uhlmann}(1970)}]{Uhlmann_1970_Shannon}%
  \BibitemOpen
  \bibfield  {author} {\bibinfo {author} {\bibfnamefont {A.}~\bibnamefont {Uhlmann}},\ }\href {\doibase 10.1016/0034-4877(70)90009-1} {\bibfield  {journal} {\bibinfo  {journal} {Reports on Mathematical Physics}\ }\textbf {\bibinfo {volume} {1}},\ \bibinfo {pages} {147} (\bibinfo {year} {1970})}\BibitemShut {NoStop}%
\end{thebibliography}

%

\appendix
\section{Proof of Theorem \ref{BDCO Theorem}}
\label{Proof of BDCO Theorem}

 Let 
    $$\mathcal{E}(X)=\sum_l K_l X K_l^\dagger$$
    be a Kraus representation of $\mathcal{E}.$ Since
    $$\mathcal{E}(\rho_\psi)=\rho_\phi.$$
    Thus each nonzero post measurement vector must be parallel to $\ket{\phi}$. Hence there exists complex numbers $c_l$ such that
    $$K_l \ket{\psi}=c_l \ket{\phi}\;\;\;\text{for all $l$}$$
    Since $\mathcal{E}$ is trace preserving,
    $$\sum_l K_l^{\dagger} K_l=I$$
    Therefore
\begin{equation*}
\begin{aligned}
    1=\bra{\psi}I\ket{\psi}&=\sum_l \bra{\psi} K_l^{\dagger} K_l \ket{\psi}\\
    &=\sum_l ||K_l \ket{\psi}||^2\\
    &=\sum_l |c_l|^2
\end{aligned}
\end{equation*}
In particular, the vector $c=(c_l)_l$ has unit norm, because
$$||c||^2=\sum_l |c_l|^2=1$$
Projecting $K_l\ket{\psi}=c_l \ket{\phi}$ onto the $\beta$th output block gives
$$\Pi_\beta K_l \ket{\psi}=c_l \Pi_\beta \ket{\phi}=c_l \ket{\phi_\beta}.$$
Since $$\ket{\psi}=\sum_\alpha \ket{\psi_\alpha},$$
Using \eqref{eq 9} We obtain for each $l$
\begin{equation}
    \sum_\alpha \Pi_\beta K_l \ket{\psi_\alpha}=c_l \sqrt{q_\beta} |\hat\phi_\beta\rangle.
    \label{eq 10}
\end{equation}
Now we use the BDCO condition. For $\alpha \neq \gamma$,
$$\Delta(\ket{\psi_\alpha}\bra{\psi_\gamma})=0.$$
Therefore
$$\Delta\mathcal{E}(\ket{\psi_\alpha}\bra{\psi_\gamma})=\mathcal{E}\Delta(\ket{\psi_\alpha}\bra{\psi_\gamma})=0$$
Thus, for every $\beta$ and $\alpha \neq \gamma$,
\begin{equation}
    \Pi_\beta \mathcal{E} (\ket{\psi_\alpha} \bra{\psi_\gamma}) \Pi_\beta=0,
\label{eq 11}
\end{equation}
Equivalently,
\begin{equation}
   \sum_l \Pi_\beta K_l \ket{\psi_\alpha} \bra{\psi_\gamma} K_l^\dagger \Pi_\beta=0
\label{eq 12}
\end{equation}
Similarly, for fixed $\alpha$,
$$\Delta(\ket{\psi_\alpha}\bra{\psi_\alpha})=\ket{\psi_\alpha}\bra{\psi_\alpha}$$
Hence
$$\Delta\mathcal{E}(\ket{\psi_\alpha}\bra{\psi_\alpha})=\mathcal{E}(\ket{\psi_\alpha} \bra{\psi_\alpha})$$
Therefore, for $\beta \neq \delta,$
\begin{equation}
    \Pi_\beta \mathcal{E} (\ket{\psi_\alpha} \bra{\psi_\alpha}) \Pi_\delta=0
\label{eq 13}
\end{equation}
Equivalently,
\begin{equation}
    \sum_l \Pi_\beta K_l \ket{\psi_\alpha} \bra{\psi_\alpha} K_l^\dagger \Pi_\delta=0
\label{eq 14}
\end{equation}
Next we define 
$$A_{\beta \alpha}=\Pi_\beta \mathcal{E} (\ket{\psi_\alpha} \bra{\psi_\alpha}) \Pi_\beta$$
Then $$A_{\beta \alpha} \geq 0$$
Also,
$$\ket{\psi}\bra{\psi}=\sum_{\alpha,\gamma} \ket{\psi_\alpha} \bra{\psi_\gamma}.$$
Using \eqref{eq 11} we can write
$$\sum_\alpha A_{\beta \alpha}=\sum_\alpha \Pi_\beta \mathcal{E} (\ket{\psi_\alpha} \bra{\psi_\alpha}) \Pi_\beta=\ket{\phi_\beta} \bra{\phi_\beta}=q_\beta |\hat\phi_\beta\rangle  \langle\hat\phi_\beta|.$$
Thus
\begin{equation}
    \sum_\alpha A_{\beta \alpha}=q_\beta |\hat\phi_\beta\rangle  \langle\hat\phi_\beta|.
\label{eq 15}
\end{equation}
Since each $A_{\beta \alpha} \geq 0$ and their sum lies in the subspace $\text{span}\{|\hat\phi_\beta\rangle\}$, each vector
$$\Pi_\beta K_l \ket{\psi_\alpha}$$
must also lie in this span. Therefore, for every $l, \beta, \alpha$ there exists a scalar $a_l(\beta, \alpha)$ such that
\begin{equation}
    \Pi_\beta K_l \ket{\psi_\alpha}=a_l(\beta ,\alpha) |\hat \phi_\beta \rangle.
\label{eq 16}
\end{equation}
We define $a_{\beta \alpha}=(a_l(\beta,\alpha))_l.$ 
Using \eqref{eq 11} and \eqref{eq 16}, we have for $\alpha \neq \gamma$
$$\sum_l \overline{a_l(\beta,\gamma)} a_l(\beta,\alpha)  |\hat\phi_\beta\rangle  \langle\hat\phi_\beta|=0.$$
Equivalently for $\alpha \neq \gamma$
\begin{equation}
    \langle a_{\beta \gamma},a_{\beta \alpha} \rangle=\sum_l \overline{a_l(\beta,\gamma)} a_l(\beta,\alpha)=0
\label{eq 17}
\end{equation}
Similarly using \eqref{eq 14} and \eqref{eq 16} we have for $\beta \neq \delta$
$$\sum_l \overline{a_l(\delta,\alpha)} a_l(\beta,\alpha) |\hat\phi_\beta\rangle  \langle\hat\phi_\delta|=0.$$
Equivalently for $\beta \neq \delta$
\begin{equation}
    \langle a_{\delta \alpha},a_{\beta \alpha} \rangle=\sum_l \overline{a_l(\delta,\alpha)} a_l(\beta,\alpha)=0
\label{eq 18}
\end{equation}
Now from \eqref{eq 10} and \eqref{eq 16},
\begin{equation}
   \sum_\alpha \Pi_\beta K_l \ket{\psi_\alpha}=\sum_\alpha a_l(\beta ,\alpha) |\hat\phi_\beta\rangle=c_l \sqrt{q_\beta} |\hat\phi_\beta\rangle
\label{eq 19}
\end{equation}
Since $|\hat\phi_\beta\rangle \neq 0,$
$$\sum_\alpha a_l(\beta , \alpha)=c_l \sqrt{q_\beta}$$
Since the above equation holds for every $l$, we have
\begin{equation}
    \sum_\alpha a_{\beta \alpha}=c \sqrt{q_\beta} .
\label{eq 20}
\end{equation}
Taking the inner product of \eqref{eq 20} with $a_{\beta \alpha}$,
\begin{equation}
\langle a_{\beta \alpha},\sum_\gamma a_{\beta \gamma} \rangle= \sqrt{q_\beta} \langle a_{\beta \alpha}, c \rangle.
\label{eq 21}
\end{equation}
For $a_{\beta \alpha} \neq 0$, we define
\begin{equation}
    \hat a_{\beta \alpha}=\frac{a_{\beta \alpha}}{||a_{\beta \alpha}||}
\label{eq 22}
\end{equation}
Using \eqref{eq 21}, \eqref{eq 22} we can write
\begin{equation*}
    \langle a_{\beta \alpha},\sum_\gamma a_{\beta \gamma} \rangle= \sqrt{q_\beta} ||a_{\beta \alpha}|| \langle \hat a_{\beta \alpha}, c \rangle.
\end{equation*}
Now we use \eqref{eq 17}. Since $||a_{\beta \alpha}|| \geq 0, \sqrt{q_\beta} > 0$, taking modulus on both sides of above equation we get
\begin{equation}
    ||a_{\beta \alpha}||^2= q_\beta | \langle \hat a_{\beta \alpha}, c \rangle |^2.
\label{eq 23}
\end{equation}
Next we define 
$$T_{\alpha \beta}=| \langle \hat a_{\beta \alpha}, c \rangle |^2.$$
Therefore,
\begin{equation}
    ||a_{\beta \alpha}||^2= T_{\alpha \beta} q_\beta.
\label{eq 24}
\end{equation}
Now by trace preserving condition,
$$p_\alpha=||\ket{\psi_\alpha}||^2=\sum_l ||K_l \ket{\psi_\alpha} ||^2.$$
Using
$$I=\sum_\beta \Pi_\beta,$$
we get
$$p_\alpha=\sum_{l,\beta} ||\Pi_\beta K_l \ket{\psi_\alpha} ||^2.$$
Using \eqref{eq 16},
$$||\Pi_\beta K_l \ket{\psi_\alpha} ||^2=|a_l (\beta, \alpha) |^2.$$
Thus
$$p_\alpha = \sum_\beta ||a_{\beta \alpha}||^2.$$
Using \eqref{eq 24},
\begin{equation}
    p_\alpha=\sum_\beta T_{\alpha \beta}q_\beta.
\label{eq 25}
\end{equation}
Hence
\begin{equation}
    p=Tq.
\label{eq 26}
\end{equation}
It remains to show that $T$ is a doubly sub-stochastic matrix. We first fix $\beta$. By \eqref{eq 17}, the normalized vectors $$\{\hat a_{\beta \alpha}\}_\alpha $$ are orthonormal. Therefore, by Bessel's inequality,
\begin{equation}
    \sum_\alpha T_{\alpha \beta}=\sum_\alpha | \langle \hat a_{\beta \alpha}, c \rangle |^2 \leq ||c||^2=1
\label{eq 27}
\end{equation}
Similarly, fixing $\alpha$, by \eqref{eq 18}, the normalized vectors $$\{\hat a_{\beta \alpha}\}_\beta $$
are orthonormal. Again by Bessel's inequality,
\begin{equation}
    \sum_\beta T_{\alpha \beta}=\sum_\beta | \langle \hat a_{\beta \alpha}, c \rangle |^2 \leq ||c||^2=1.
\label{eq 28}
\end{equation}
Since also
$$T_{\alpha \beta} \geq 0,$$
using \eqref{eq 27} and \eqref{eq 28} we can say that the matrix $T$ is doubly sub-stochastic matrix.\\
Since $T$ is doubly sub-stochastic and $\tilde{p}=T \tilde{q}$ holds, Theorem $2$ gives $\tilde{p} \prec_w \tilde{q}.$ Since
$$\sum_\alpha p_\alpha =\sum_\beta q_\beta=1,$$
Weak majorization becomes strong majorization.
hence 
\begin{equation}
    \tilde{p} \prec \tilde{q}.
\label{eq 29}
\end{equation}
This completes the necessary part.

Conversely, let us consider $\tilde{p} \prec \tilde{q},$ then there exists a doubly stochastic matrix $T$ such that 
$$\tilde{p}=T \tilde{q}.$$
By the Birkhoff-Von Neumann theorem, we can write
$$T=\sum_l \lambda_l S_l,$$
where $\lambda_l \geq 0,$ $\sum_l \lambda_l=1,$ and each $S_l$ is a permutation matrix. Let $\pi_l$ denote the permutation associated with $S_l,$ so that 
$$(S_l \tilde{q})_\alpha=q_{\pi_l (\alpha)}.$$
Equivalently for each $\alpha$
\begin{equation}
p_\alpha = \sum_l \lambda_l q_{\pi_l(\alpha)}
\label{eq 30}
\end{equation}
For each $l$, we define
\begin{equation}
    A_l=\sqrt{\lambda_l} \sum_\alpha \sqrt{\frac{q_{\pi_l(\alpha)}}{p_\alpha}} |\hat\phi_{\pi_l(\alpha)} \rangle \langle \hat \psi_\alpha |
\label{eq 31}
\end{equation}
From \eqref{eq 31} we can write
\begin{equation*}
    \begin{aligned}
        A_l \ket{\psi} &= \left( \sqrt{\lambda_l} \sum_\alpha \sqrt{\frac{q_{\pi_l(\alpha)}}{p_\alpha}} |\hat\phi_{\pi_l(\alpha)} \rangle \langle \hat \psi_\alpha | \right) \left(\sum_\gamma \sqrt{p_\gamma} | \hat\psi_\gamma \rangle \right)\\
        &=\sqrt{\lambda_l} \sum_\alpha \langle \hat\psi_\alpha | \sum_\gamma \sqrt{p_\gamma} |\hat\psi_\gamma \rangle \sqrt{\frac{q_{\pi_l(\alpha)}}{p_\alpha}} |\hat\phi_{\pi_l(\alpha)} \rangle \\
    \end{aligned}
\end{equation*}

Using the relation $\langle \hat\psi_\alpha | \hat\psi_\gamma \rangle=\delta_{\alpha \gamma}$ we can write the above equation as

\begin{equation}
    A_l \ket{\psi}= \sqrt{\lambda_l} \sum_\alpha \sqrt{q_{\pi_l(\alpha)}} |\hat\phi_{\pi_l(\alpha)} \rangle 
\label{eq 32}
\end{equation}
Here $\alpha$ runs over all input block levels $\alpha=1,2,...,m$ and since $\pi_l$ is permutation, the values $\pi_l(\alpha)$ also runs over block levels $1,2,..m$ exactly once. so we can write $\beta=\pi_l(\alpha)$. so we can define
$$\beta=\pi_l(\alpha)$$
From \eqref{eq 32}
\begin{equation}
    \begin{aligned}
        A_l \ket{\psi}&=\sqrt{\lambda_l} \sum_\alpha \sqrt{q_{\pi_l(\alpha)}} |\hat\phi_{\pi_l(\alpha)} \rangle \\
        &= \sqrt{\lambda_l}  \sum_\beta \sqrt{q_\beta} | \hat\phi_\beta \rangle\\
        &= \sqrt{\lambda_l} \ket{\phi}
    \end{aligned}
\label{old kraus}
\end{equation}
Consequently
$$\mathcal{E} (\ket{\psi} \bra{\psi})=\ket{\phi} \bra{\phi}$$ holds.\\
But in general,
\begin{equation}
\begin{aligned}
    \sum_l A_l^\dagger A_l &= \sum_\alpha | \hat \psi_\alpha \rangle \langle \hat \psi_\alpha |\\
    & \neq I.
\end{aligned}
\end{equation}

So completeness condition is not satisfied for this set of Kraus operators.
To make the map trace preserving on the whole Hilbert space, we first chose an orthonormal basis for each block $\mathcal{H_\alpha}$ 
$$\{ |\hat\psi_\alpha \rangle, | e_1^{(\alpha)} \rangle, |e_2^{(\alpha)} \rangle,....,| e_{d_\alpha-1}^{(\alpha)} \rangle \},$$
Where $dim \mathcal{H_\alpha}=d_\alpha.$


For every output block $\mathcal{H}_{\pi_k(\alpha)}$, we choose an normalized vector $|f_{\pi_k(\alpha)} \rangle \in \mathcal{H}_{\pi_k(\alpha)}.$

Now we consider full Kraus family $K_l$ as  $\{A_l\}_l$ $\cup$ $\{R_{l, \alpha, r}\}_{l, \alpha,r},$
Where \\
$$A_l= \sqrt{\lambda_l} \sum_\alpha \sqrt{\frac{q_{\pi_l(\alpha)}}{p_\alpha}} |\hat\phi_{\pi_l(\alpha)} \rangle \langle \hat \psi_\alpha |,$$
and 
$$R_{l, \alpha, r}= \sqrt{\lambda_l} |f_{\pi_l(\alpha)} \rangle \langle e^{(\alpha)}_r|.$$
Hence action of $\mathcal{E}$ on some state $\rho$ will be
\begin{equation}
    \mathcal{E} (\rho)= \sum_l A_l \rho A_l^\dagger + \sum_{l, \alpha, r} R_{l, \alpha, r} \rho R_{l, \alpha, r}^\dagger
\label{eq 42}
\end{equation}
Now 
\begin{equation}
    \begin{aligned}
        \sum_l R_{l , \alpha, r}^\dagger R_{l, \alpha,r} &= \sum_l \lambda_l | e^{(\alpha)}_r \rangle \langle e^{(\alpha)}_r|\\
        &= | e^{(\alpha)}_r \rangle \langle e^{(\alpha)}_r|.
    \end{aligned}
\end{equation}

Summing over all $\alpha, r$ we have
$$\sum_{l, \alpha, r} R_{l, \alpha, r}^\dagger R_{l,\alpha,r}=\sum_\alpha \sum_r | e^{(\alpha)}_r \rangle \langle e^{(\alpha)}_r|=I-\sum_\alpha | \hat \psi_\alpha \rangle \langle \hat \psi_\alpha | $$

Using the above equation
$$\sum_l A_l^\dagger A_l+\sum_{l, \alpha, r} R_{l, \alpha, r}^\dagger R_{l, \alpha, r}=I$$
So, the completeness condition is satisfied
Now 

Now by the construction of orthonormal basis we have
$$ \langle e_r^{(\alpha)} | \psi \rangle = \langle e_r^{(\alpha)} | \sum_\alpha \sqrt{p_\alpha} |\hat\psi_\alpha \rangle \rangle=0,$$

Therefore,
$$R_{l, \alpha, r} \ket{\psi}=0.$$

So only the operators $A_l$ contribute to the transformation from $\ket{\psi} \bra{\psi}$ to $\ket{\phi} \bra{\phi}$.

Using above expression and \eqref{old kraus} 
\begin{equation}
\begin{split}
\mathcal{E} (\ket{\psi} \bra{\psi}) &= \sum_l A_l (\ket{\psi} \bra{\psi})  A_l^\dagger\\ 
&+ \sum_{l, \alpha, r} R_{l, \alpha, r} (\ket{\psi} \bra{\psi})R_{l, \alpha, r}^\dagger\\
&= \ket{\phi} \bra{\phi}
\end{split}
\end{equation}

So, the desired state transformation result holds using extended form of Kraus operator.\\

It remains to show that $\mathcal{E}$ is BDCO. Equivalently
$$\Delta \circ \mathcal{E}= \mathcal{E} \circ \Delta$$

From the construction of $A_l$ it is obvious that
$$A_l \Pi_\alpha X \Pi_\gamma A_l^\dagger$$
is supported between the output blocks $\mathcal{H}_{\pi_l(\alpha)}$ and $\mathcal{H}_{\pi_l(\gamma)}$ and is therefore annihilated by $\Delta$ for $\alpha \neq \gamma$. On the other hand, the terms with $\alpha = \gamma$ are supported inside a single output block and are left invariant by $\Delta$. Hence
\begin{equation}
\Delta(A_l X A_l^\dagger)=\sum_\alpha A_l \Pi_\alpha X \Pi_\alpha A_l^\dagger=A_l \Delta(X) A_l^\dagger
\label{eq 45}
\end{equation}


Similarly, from construction of $R_{l, \alpha, r}$ we can claim that 
$$R_{l, \alpha, r}=\Pi_{\pi_l(\alpha)} R_{l, \alpha, r} \Pi_\alpha.$$
Therefore,
$$R_{l, \alpha, r} X R_{l, \alpha, r}^\dagger=R_{l, \alpha, r} \Pi_\alpha X \Pi_\alpha R_{l, \alpha, r}^\dagger,$$
and the output is supported inside the single block $\mathcal{H}_{\pi_l(\alpha)}$. Thus
\begin{equation}
\Delta (R_{l, \alpha, r} X R_{l, \alpha, r}^\dagger)=R_{l, \alpha, r} \Delta(X) R_{l, \alpha, r}^\dagger
\label{eq 46}
\end{equation}
Finally using \eqref{eq 42}
$$\Delta(\mathcal{E}(X))=\mathcal{E}(\Delta(X)).$$
This completes the proof.

\section{Proof of Corollary \ref{SBIO corollary} }
\label{Appendix B}

The necessity part follows from the BDCO result which we have already proved. Since every SBIO channel belongs to the BDCO class, any pure state transformation achievable by SBIO must satisfy the same necessary condition as BDCO. Therefore, if an SBIO channel $\mathcal{E}$ maps $\ket{\psi} \bra{\psi}$ to $\ket{\phi} \bra{\phi}$, then the block vectors must satisfy $p \prec q.$\\
    For sufficiency we assume $p \prec q.$ By the construction in the converse part of Theorem 4, the Kraus family is $\{A_l\}_l$ $\cup$ $\{R_{l, \alpha, r}\}_{l, \alpha, r}$. From this construction we can conclude that each Kraus operator maps every input block into a single output block. Consequently, for every $X$,
    \eqref{eq 45} and \eqref{eq 46} gives
    $$\Delta(A_l X A_l^\dagger)=A_l \Delta(X) A_l^\dagger,$$
    $$\Delta(R_{l, \alpha, r} X R_{l, \alpha, r}^\dagger)=R_{l, \alpha, r} \Delta(X) R_{l, \alpha, r}^\dagger.$$
    Therefore constructed channel $\mathcal{E}$ is SBIO. This proves the sufficiency.

\end{document}